\documentclass[11pt]{amsart}  
\usepackage{amsmath,amsrefs,mathtools}
\usepackage{mcite}
\usepackage{amssymb}
\usepackage{harpoon}
\usepackage{subfig}
\usepackage{float}
\usepackage{bm}
\usepackage{enumerate}
\usepackage[dvips]{epsfig}
\usepackage{enumitem}
\usepackage{amsmath}
\usepackage{graphicx}
\usepackage{pgfplots}
\usepackage{bm}
\usepackage{hyperref} 
\hypersetup{colorlinks=true, linkcolor=red, citecolor = blue, urlcolor=blue, filecolor=magenta}
\newcommand{\di}{\displaystyle}
\usepackage{subfig,tikz}
\usetikzlibrary{decorations.markings}
\usepackage{wrapfig}
\usepackage{xspace}
\theoremstyle{plain}
\newtheorem{theorem}{Theorem}[section]
\newtheorem{lemma}[theorem]{Lemma}
\newtheorem{claim}[theorem]{Claim}
\newtheorem{prop}[theorem]{Proposition}
\newtheorem{cor}[theorem]{Corollary}

\newcommand{\norm}[1]{\left\lVert#1\right\rVert}
\newtheorem{remark}[theorem]{Remark}

\newcommand{\td}{\mbox{d}}
\newcommand{\bx}{\mathbf{x}}
\newcommand{\ti}{{\hat{i}}}
\newcommand{\tj}{{\hat{j}}}
\numberwithin{equation}{section}

\DeclarePairedDelimiter\floor{\lfloor}{\rfloor}

\allowdisplaybreaks
\begin{document}

\title[Slow decay of waves in gravitational solitons]
{Slow decay of waves in gravitational solitons}
\author[Sharmila Gunasekaran]{Sharmila Gunasekaran}
\address{Department of Mathematics and Statistics, Memorial University of Newfoundland,
	St. John's, NL A1C5S7, Canada}
\email{sdgg82@mun.ca}
\author[Hari K. Kunduri]{Hari K. Kunduri}		
\address{Department of Mathematics and Statistics, Memorial University of Newfoundland,
	St. John's, NL A1C5S7, Canada}
\email{hkkunduri@mun.ca}

\begin{abstract} We consider a family of globally stationary (horizonless), asymptotically flat solutions of five-dimensional supergravity. We prove that massless linear scalar waves in such soliton spacetimes cannot have a uniform decay rate faster than inverse logarithmically in time. This slow decay can be attributed to the stable trapping of null geodesics. Our proof uses the construction of quasimodes which are time periodic approximate solutions to the wave equation. The proof is based on previous work to prove an analogous result in $\mbox{Kerr-AdS}_4$ black holes \cite{holzegel:2013kna}. We remark that this slow decay is suggestive of an instability at the nonlinear level. 
\end{abstract}

\maketitle

\section{Introduction}
Gravitational solitons are globally stationary, asymptotically flat spacetimes with positive energy.  A classic result of Lichnerowicz \cite{Lich} demonstrates that there are no such vacuum solutions in four dimensions.  The result can be obtained more directly from the modern viewpoint by an application of the positive mass theorem along with Stokes' theorem and identities related to the stationary Killing field. Intuitively, the result states that an isolated self gravitating system in equilibrium with positive energy must contain a black hole~\cite{Gibbons:1997cc}. The result extends to Einstein-Maxwell theory and vacuum general relativity in dimensions greater than four.  However, within the supergravity theories that govern the low-energy  dynamics in string theory,  gravitational solitons arise naturally (we note that \emph{static} solitons can be ruled out in pure Einstein-Maxwell theory in $D > 4$ \cite{Kunduri:2017htl}, and there are no known stationary examples).  In fact several large families of such supergravity solutions have been obtained explicitly (e.g. see the review \cite{Bena:2007kg}). The solitons obtained in these constructions are typically characterized by their mass, angular momenta, global electric charges, and non-trivial spacetime topology.  They have received considerable interest, as it has been suggested that they represent classical `microstate geometries'  corresponding to black holes carrying the same conserved charges, thus providing a resolution to the information paradox~\cite{Mathur:2005zp}.

Quite independently of these considerations, gravitational solitons possess a number of novel features that distinguish them from black holes.  In particular, certain supersymmetric examples contain `evanescent ergosurfaces', which are timelike hypersurfaces upon which the stationary Killing field can become null \cite{Gibbons:2013tqa}. It has been proved that such spacetimes suffer from nonlinear instabilities~\cite{Eperon:2016cdd, Keir:2016azt} and exhibit a certain kind of linear instability \cite{Keir:2018hnv}.  On the other hand, soliton spacetimes satisfy a mass variation formula which is analogous to the familiar first law of black hole mechanics~\cite{Kunduri:2013vka}. Moreover, solutions have been explicitly constructed that physically correspond to bound state configurations of black holes and solitons (i.e. they have 2-cycles in the domain of outer communication)~\cite{Kunduri:2014iga, Breunholder:2017ubu}. Somewhat surprisingly, these solutions have been shown to lead to a continuous failure of black hole uniqueness in higher dimensions even in the supersymmetric setting~\cite{Horowitz:2017fyg}.

A natural question to consider is whether these globally stationary solutions are actually stable in some precise sense.  There is, of course, presently a rich body of results concerning the analogous problem for stationary black holes. This stability problem can be posed at increasing levels of complexity.  As is well known, the  Einstein equations in a suitable gauge reduce to the following schematic form, 
\begin{align}
\Box_g g_{\mu \nu} = Q_{\mu \nu}(g,\partial g) + T_{\mu \nu}
\label{ee1}
\end{align}
where $Q$ is quadratic in $\partial g$. One of the important questions concerning explicit solutions to \eqref{ee1} is the analysis of their nonlinear stability in a similar vein as the groundbreaking work of Christodoulou-Klainerman \cite{Christodoulou:1993uv} \footnote{Alternate proofs for this nonlinear stability
	result have been obtained in \cite{lindblad:2004ue} and \cite{hintz:2017xxu}.}. In this work, it was made clear that perturbations  propagate as waves. A natural associated problem to consider is the coupled set of equations governing gravitational perturbations, namely those obtained by linearizing \eqref{ee1} about an explicit solution.  Hence the equation for a single massless scalar field $\Phi$, in a fixed background $(M,g)$, is a good starting point:
\begin{equation}
\Box_g \Phi = 0.
\label{lwave_eq}
\end{equation} 
Though \eqref{lwave_eq} is the simplest version of  the gravitational perturbation equations, it still preserves many geometric features of the spacetime through the metric, $g$. Hence understanding the properties of solutions to \eqref{lwave_eq} is a useful precursor to the problem of nonlinear stability in the spacetime.
The study of linear scalar waves in spacetimes has a well established history; \cite{Dafermos:2008en,dafermos:2010hd,Andersson:2016hmv,Finster:2008bg,Finster:2018zfe} are essential reviews on the subject. The study of linear wave equations on explicit stationary solutions has also seen remarkable advancements for spacetimes with other asymptotics and dimensions greater than four. We present a non-exhaustive review below with a marked focus on the methods and results most pertinent to the present work. Our work falls under the domain of stability results in stationary asymptotically flat backgrounds in five spacetime dimensions. 

In the realm of stationary asymptotically flat black hole spacetimes, two central unresolved problems are to confirm or disprove the nonlinear stability of the Schwarzschild and Kerr solutions. The initial investigations into stability were focused on mode analysis which confirms the absence of  certain exponentially growing modes (in the subextremal case for Kerr)~\cite{Whiting:1988vc,Regge:1957td}. However these results do not address any boundedness or decay of perturbations. The first step in this direction was the proof of boundedness of scalar waves on Schwarzschild spacetime by Kay--Wald \cite{Kaywald87,Wald79} with stronger results subsequently obtained using more universal and robust techniques \cite{Dafermos:2008en,Blue:2003pj,Blue:2003si,Blue:2006nj,2000math.ph2030L,Dafermos:2003yw,2006CMaPh.268..481B}. The black hole case presents a number of challenges, most notably the  degeneracy of energy at the horizon, the trapping of null geodesics \cite{Sbierski:2013mva} and superradiance. These problems have been addressed with significant progress in quantitative decay rates \cite{Dafermos:2007ab,Dafermos:2009uq,dafermos:2010hb,Shlapentokh-Rothman:2013hza}. These efforts culminated in the proof for decay of linear waves in sub-extremal Kerr spacetime by Dafermos--Rodnianski--Shlapentokh-Rothman \cite{Dafermos:2014cua}  (see also \cite{2008arXiv0810.5766T,Andersson:2009jr,Finster:2005dg}). In  contrast, extremal black holes are affected by an instability discovered by Aretakis (non-decay along the horizon) which also affects long-time decay as discussed in \cite{Aretakis:2018dzy, Angelopoulos:2018uwb, Aretakis:2012ei}. This also implies that the extremal Kerr solution is unstable to linearized gravitational perturbations as shown by Lucietti--Reall in \cite{Lucietti:2012sf}. For the Schwarzschild case, linear stability under the full set of gravitational perturbations (i.e. the linearization of \eqref{ee1})  was proved by Dafermos--Holzegel--Rodnianski \cite{Dafermos:2016uzj} (see also \cite{andersson2019stability}). It is now known due to Klainerman--Szeftel that Schwarzschild is \emph{nonlinearly} stable to the class of polarized perturbations~\cite{Klainerman:2017nrb}. See  \cite{Dafermostalk2019,Holzegeltalk2020} for the recent annoucement of the full finite-codimension non-linear asymptotic stability of the Schwarzschild family. The authors of~\cite{Dafermos:2016uzj} have further established boundedness and polynomial decay for the spin-2 Teukolsky equation on the Kerr spacetime, which is required to prove the full linearized stability of Kerr to gravitational perturbations~\cite{Dafermos:2017yrz}. Hafner--Hintz--Vasy in \cite{Hafner:2019kov} proved linear stability for slowly rotating Kerr black holes using spectral methods. 

One may consider stability problems that are asymptotically Anti-de Sitter (AdS) or de Sitter (dS) which are the two other maximally symmetric constant curvature backgrounds. In particular, vacuum  AdS, which has a timelike boundary, has been conjectured to be unstable under perturbations of its initial data leading to the formation of a black hole. Numerical work strongly supporting this claim was given in the seminal work of Bizon--Rostworowski \cite{bizon:2011gg}. The rigorous results by Moschidis \cite{moschidis:2017llu,moschidis:2018ruk} give further strong evidence for the instability. Recent progress in this problem was announced in \cite{Moschidistalk2020}. The decay of Klein Gordon fields in AdS was investigated in \cite{Holzegel:2012wt,Warnick:2012fi,Warnick:2013hba} and the global dynamics of solutions to the massive wave equation in AdS black hole spacetimes have been investigated in  \cite{Holzegel:2012wt,holzegel:2011uu}. In particular as we discuss below, they exhibit a slow decay rate.  Finally, on general asymptotically dS spacetimes, waves decay exponentially fast, in contrast with the asymptotically flat case where the decay is at most polynomial. For results on the nonlinear stability of the dS spacetime see \cite{Anderson:2004ir,friedrich1986} with extensions in \cite{Ringstrm2008FutureSO,Rodnianski:2009de}.  Remarkably,  the nonlinear stability of slowly rotating Kerr-dS spacetime has been proved by Hintz--Vasy in \cite{hintz:2016gwb} and extended by Hintz in \cite{Hintz:2016jak}.

The investigation of  the stability for higher-dimensional black holes has also received much recent attention. The problem is motivated both for intrinsic mathematical reasons and by connections to high energy physics (see the review \cite{Emparan:2008eg}).  Unsurprisingly, the presence of extra spatial dimensions allows for various novel geometric and topological features, such as the gravitational solitons discussed here and black holes with non-spherical topology. An important rigorous result is that of Schlue, who proved robust quantitative energy decay estimates for solutions of~\eqref{lwave_eq} in the Schwarzschild family in $D>4$ spacetime dimensions (see also \cite{Dafermos:2005nh, Holzegel:2008ve, laul2010localized,laul:2014gla}).  There is a rather vast literature on mode instabilities associated to rotating Myers-Perry black holes (the natural generalization of the Kerr solution) that arise at sufficiently high angular momenta, as well as numerical analyses on the dynamical evolution \cite{Figueras:2017zwa,Dias:2010eu,Bantilan:2019bvf}.  Like the Kerr solution, in the Myers-Perry background, \eqref{lwave_eq} admits separable solutions, which is  particularly useful in the above studies.  The black ring family of solutions  that describe rotating, asymptotically flat black holes with $S^1 \times S^2$ topology \cite{Emparan:2001wn,Pomeransky:2006bd}, in contrast, are not presented in coordinates which admit a similar separation of variables.  This has impeded progress on stability outside of robust numerical strategies \cite{Figueras:2015hkb}. Nonetheless, Benomio \cite{benomio:2018ivy} has recently proved that the uniform decay rate is slow for generic solutions to \eqref{lwave_eq}.  This provides strong evidence that black rings must be nonlinearly unstable. Quite recently, the nonlinear stability of higher dimensional spacetimes that arise in supersymmetric compactifications of string theory was investigated in \cite{Andersson:2020fuz,Wyatt:2017tow}.

One of the main geometric obstructions to proving a strong decay statement (i.e., fast decay) for solutions to \eqref{lwave_eq} is the phenomena of \emph{trapping} - the confinement of null geodesics in a bounded region of space.  The rates of decay of solutions to \eqref{lwave_eq} are characterized as \emph{fast} or \emph{slow} depending on their applicability in nonlinear problems. Polynomial decay is robust enough to give hope for nonlinear stability whereas logarithmic decay is not and is hence considered slow. A well-known example of trapping occurs at the photon sphere ($r = 3M$) of Schwarzschild spacetime. Here, initially trapped geodesics are not trapped when perturbed and this structure is characterized as \emph{unstable trapping}. The trapping in Kerr black holes is another such example.  Since the propagation of high-frequency waves can be approximated by null geodesics, intuitively one expects energy to clump in a trapped region, leading to slower decay.  When trapping is the only obstruction, how strongly the geodesics are trapped is a factor that ultimately dictates whether there is slow or fast decay. The unstable trapping in the Schwarzschild solution roughly leads to sufficiently fast decay.  In contrast, the structure of trapping in the soliton geometry to be considered here is \emph{stable}. 

The question of whether waves decay at all was answered in the affirmative for a general class of stationary asymptotically flat spacetimes due to a powerful result of Moschidis \cite{Moschidis:2015wya}. The general decay result he established is restated here : 
\begin{theorem}[Moschidis, 2015] \label{Moschidis}	
	\label{logupb}
	Let
	$({\mathcal{M}}^{d+1}, g) $, $d \geq 3$, be a globally hyperbolic spacetime, which is stationary and asymptotically flat, and
	which can possibly contain black holes with a non-degenerate 	horizon and a small ergoregion. Moreover, suppose that an energy boundedness statement is true for solutions $\Phi$ of the linear wave equation \eqref{lwave_eq} on the domain of outer communications ${\mathcal{D}}$ of the spacetime. Then the local energy of $\Phi$ on ${\mathcal{D}}$ decays at least with a logarithmic rate :
	\begin{equation}\label{upperbound}
	E_{loc}[\Phi](t) \leq C_m \frac{1}{\{\log(2 + t) \}^{2m}} E^m_w[\Phi](0)
	\end{equation}
	where $t$ is a suitable time function on ${\mathcal{D}}$ and $E^m_w[\Phi](0)$ is an initial energy based on the first $m$ derivatives of $\Phi $.
\end{theorem} 
\noindent We note in particular that the above result establishes an upper bound on decay for solutions to \eqref{lwave_eq} for a wide class of spacetimes (that is, it is a statement asserting that waves must decay at least inverse logarithmically). 

The results in this paper, following closely the strategy of  \cite{holzegel:2013kna,Keir:2016azt,Keir:2014oka,benomio:2018ivy} follow from an investigation of slow decay rates for certain stationary spacetimes. In  these spacetimes, there are families of trapped null geodesics that have the property that perturbed null geodesics will still be trapped. Hence this strcuture of trapping is \emph{stable}. Examples of spacetimes exhibiting stable trapping are Kerr-AdS$_4$ black holes \cite{holzegel:2013kna}, ultracompact neutron stars \cite{Keir:2014oka}, black strings and black rings (mentioned above) \cite{benomio:2018ivy} and the supersymmetric \emph{microstate geometries} of \cite{Keir:2016azt}. We recall that microstates are stationary, asymptotically flat horizonless solutions of supergravity, and hence in our terminology above, are examples of gravitational solitons. In fact, in contrast to the non-supersymmetric solitons we consider, the solutions  investigated in \cite{Keir:2016azt}  possess either an ergoregion or evanescent ergosurface.  Physically, the mechanism of stable trapping at work in Kerr-AdS$_4$ black holes is the combined effect of lack of dispersion at the asymptotic end and the usual unstable trapping outside the horizon \cite{holzegel:2011uu}, whereas in the case of ultracompact neutron stars and microstates,  stable trapping is a result of the coupling between the lack of horizon and trapping.  The mechanism behind stable trapping for black rings appears related to the topology of the domain of outer communication. The slow decay result pertaining to stable trapping in supersymmetric solitons  proved in \cite{Keir:2016azt} is clearly most relevant to our problem, and is restated here:
\begin{theorem}[Keir, 2017]
	Let $\Phi$ be a solution to the wave equation \eqref{lwave_eq} on a two-charge geometry. Let $\Omega$ be an open set containing the trapped region. Then for all $k\geq 1$, there exist positive constants $C_k$ such that, 
	\begin{equation}
	\limsup_{t \to \infty} \sup_{\Phi \neq 0} \left( \frac{\log (2 + t)}{\log \log( (2 + t))} \right)^{2k} \frac{E_\Omega[\Phi](t)} {E_{k+1} [ \Phi](0)} \geq C_k. 
	\end{equation}
\end{theorem}
\noindent We remark here that these solutions are the `closest  analogue' to extremal black holes for horizonless solutions bearing in mind that there is no notion of surface gravity here. There are some similarities to extremal black holes with regards to the kind of instability these solutions exhibit as noted by Keir in \cite{Keir:2018hnv} - namely, that the solutions to the linear wave equation in these backgrounds have a quantity that is non-decaying on a particular surface, 
though in this case the surface is an evanescent hypersurface which is timelike.

In contrast to the family of solutions studied in~\cite{Keir:2016azt}, the soliton spacetimes we examine are non-supersymmetric, devoid of an ergoregion or evanescent ergosurface,  and the energy of solutions to the massless wave equation i.e., \eqref{lwave_eq} is easily seen to be  uniformly bounded. Hence unlike \cite{Keir:2016azt} and the solutions discussed in \cite{Cardoso:2005gj}, the spacetime  we investigate satisfies the conditions for the application of the upper bound stated in Theorem \ref{logupb}. In this paper we prove a \emph{lower bound} for the decay rate.  More precisely, our main result is 
\begin{theorem}
	\label{key_theorem}
	Let $\Phi$ be a solution to the wave equation \eqref{lwave_eq} on a soliton spacetime. Let $\Omega$ be an open set containing the trapped region. Then for all $k\geq 1$, there exist positive constants $C_k$ such that,
	\begin{equation}
	\limsup_{t \to \infty} \sup_{\Phi \neq 0} (\log (2 + t))^{2k} \frac{E_\Omega[\Phi](t)} {E_{k+1} [ \Phi](0)} \geq C_k 
	\end{equation}
\end{theorem} \noindent where the supremum is taken over all functions $\Phi$ in the completion of the set of smooth, compactly supported functions with respect to the norm defined by the higher order energy, $E_{k+1}$. See \eqref{energy} and \eqref{higherenergy} for the definition of energy.  

An immediate consequence of this result, in conjunction with Moschidis' Theorem \ref{Moschidis}, is that the bound given by \eqref{upperbound} is sharp for this class of spacetimes. Furthermore, our result strongly suggests that decay in the fully nonlinear regime is unlikely. As mentioned in \cite{Keir:2014oka}, one expects the end point of such a nonlinear instability to be gravitational collapse, intuitively caused by the trapping of waves. In light of this result, it would be interesting to study the stability of more general families of nonsupersymmetric solitons that were constructed in \cite{Jejjala:2005yu}.  Furthermore, it would be natural to extend the investigations here to investigate the stability of spacetimes containing both a black hole and soliton \cite{Kunduri:2014iga}, or a black lens \cite{Kunduri:2014kja} (an asymptotically flat black hole with horizon topology $S^3 / \mathbb{Z}_2$), which contains both a horizon and an evanescent ergosurface in the domain of outer communcations. One might expect that the presence of a horizon might influence the stability. 

This document is organized as follows. We introduce solitons and review the properties of the spacetime in \S \ref{sec:metric_prop}. We understand trapping by studying null geodesics. More specifically, we prove that, from the geodesic point of view there is a region of phase space exhibiting stable trapping. The uniform boundedness argument in this spacetime is quite straightforward and we give this in  \S \ref{sec:boundedness} to present a complete discussion on stability. In \S \ref{sepvar}, after a separation of the wave equation into a one variable Shr\"{o}dinger type equation, we  see how geodesic trapping manifests in high frequency waves. This Shr\"{o}dinger type equation is a nonlinear eigenvalue problem and establishing the existence of eigenavalues to this problem is central to proving the lower bound on the uniform decay rate. Here, we also state the nonlinear eigenvalue problem ($\mathcal{P}_\beta$) and the corresponding linear eigenvalue problem ($\mathcal{P}_0$) that will be studied first. In \S \ref{sec:linear}, $\mathcal{P}_0$ is examined and the existence of eigenvalues to this problem is proved using a version of Weyl's law. In \S \ref{sec:nonlinear}, we move to the actual nonlinear problem of interest $\mathcal{P}_\beta$. We start by examining the properties of the `nonlinear potential' and restore the setting of $\mathcal{P}_0$ for $\mathcal{P}_\beta $. Using the bounds on the eigenvalues and the implicit function theorem, we will establish the existence of eigenvalues to $\mathcal{P}_\beta$. The remaining part of the paper contains the details of how these eigenfunctions prove a logarithmic lower bound on the uniform decay rate. In \S \ref{sec:lowerbound}, we use Agmon estimates to quantitatively measure the solution (eigenfunctions) in the cut-off region. This estimate decays exponentially in a certain parameter $n$. Quasimodes are constructed by smoothly cutting off the solution near the boundary of a set containing the trapped region. The corresponding wave function $\Psi$ is shown to be an approximate solution to the wave equation \eqref{lwave_eq} with an exponentially small error in $n$. This in conjunction with Duhamel's formula will give the logarithmic lower bound. Our work is heavily indebted to the clear exposition given by Benomio  \cite{benomio:2018ivy}.

\subsection*{Acknowledgements}
 This work was supported by NSERC Discovery Grant RGPIN-2018-04887. We are thankful to Stefanos Aretakis for helpful comments on a preliminary version of this paper and for discussions during the 2017 Atlantic General Relativity workshop at Memorial University. We also acknowledge helpful communications with Joe Keir on the derivation of certain estimates in \S \ref{sec:lowerbound}.  We thank Graham Cox for clarifying various aspects of PDEs and Ivan Booth and James Lucietti for useful discussions. 

\section{A Class of Nonsupersymmetric Gravitational Solitons}
\label{sec:metric_prop}

\subsection{Metric and properties of the solution}
We consider an asymptotically flat, globally stationary family of non-supersymmetric soliton spacetimes.  The underlying manifold has topology $\mathbb{R} \times \Sigma$ with the spatial slices $ \Sigma \cong \mathbb{R}^4 \# \mathbb{CP}^2  $. It is analyzed in detail in  \cite{Gibbons:2013tqa} and \cite{Gunasekaran:2016nep}.  The spacetime  is a solution to five-dimensional minimal supergravity whose action is
\begin{equation}\label{minsugra}
S = \frac{1}{16\pi} \int_\mathcal{M} \left( \star R -2 F \wedge \star F -\frac{8}{3\sqrt{3}} F \wedge F \wedge A\right) \,.
\end{equation} Here $F = \td A$ is a smooth 2-form on the spacetime describing the Maxwell field and $A$ is a locally defined gauge potential. The local solution $(g,F)$ is
\begin{align}
\begin{aligned}
\td s^2 &= \frac{-r^2 W(r)}{ 4 b(r)^2} \td t^2 + \frac{\td r^2}{W(r)} + \frac{r^2}{4} ( \sigma_1^2 + \sigma_2^2 ) + b(r)^2 (\sigma_3 + f(r) \td t )^2 \, , \\
F & = \frac{\sqrt{3} q}{2}  \td \left[ r^{-2} \left( \frac{j}{2} \sigma_3 - \td t\right)\right]\, .
\label{solmetric}
\end{aligned}
\end{align}
The functions appearing in the metric are given below :
\begin{align}
\begin{aligned}
W(r) = 1 - \frac{2}{r^2} \left( p - q  \right) &+ \frac{1}{r^4} \left( q^2 + 2 p j^2 \right) \,,\quad
b(r)^2  =  \frac{r^2}{4} \left( 1 - \frac{j^2 q^2}{r^6} + \frac{2 j^2 p}{r^4} \right) \\
& f(r)  = \frac{-j}{2 b(r)^2} \left( \frac{2p -q}{r^2} - \frac{q^2}{r^4} \right) \, ,
\end{aligned}
\end{align} 
and the $\sigma_i$ are left-invariant one-forms on $SU(2)$ given by
\begin{align}
\begin{aligned}
{\sigma}_1 = - \sin \psi \td \theta + \cos \psi & \sin \theta \td \phi ,\,\,   {\sigma}_2 = \cos \psi \td \theta + \sin \psi \sin \theta \td \phi \\
{\sigma}_3  &= \td \psi + \cos \theta \td \phi
\end{aligned}
\end{align}
which satisfy $ \td \sigma_i = \tfrac{1}{2}\epsilon_{ijk} \sigma_j \wedge \sigma_k $.  In order to describe an asymptotically flat metric in the region $r \to \infty$, we must periodically identify $\psi \sim \psi + 4\pi, \phi \sim \phi + 2\pi$ and $\theta \in (0,\pi)$. $t \in \mathbb{R}$ is the time coordinate. The range of the radial coordinate is $0 < r_0 < r < \infty$ where $r_0$ is a parameter that characterizes the size of an $S^2$ `bolt' as described below. The paramaters $p,q,r_0$ and $j$ are related by,
\begin{align}
p &= \frac{r_0^4(r_0^2 - j^2)}{2 j^4}\,, \quad
q = \frac{-r_0^4}{j^2}
\end{align}
The spacetime is invariant under an $\mathbb{R} \times SU(2)\times U(1)$ isometry generated by $\partial_t$, $\partial_\psi$ and the vector fields $R_i$ that leave the $\sigma_i$ invariant. The above solutions are a subfamily of a more general set of $\mathbb{R} \times \times U(1)^2$ invariant nonsupersymmetric solitons (see \cite{Jejjala:2005yu,Compere:2009iy}). Surfaces of constant $r>r_0$ are timelike hypersurfaces with spatial geometry of $S^3$ with a homogeneously squashed metric.   An analysis of the metric shows that it is smooth everywhere (apart from standard coordinate singularities at $\theta = 0, \pi$ corresponding to fixed points of $U(1)$ isometries on $S^3$). However, the parameters $p$ and $q$ have been chosen above so that the functions $W(r),\, b(r)$ have simple zeroes at $r =r_0$. In particular the Killing field $\partial_\psi$ degenerates at $r_0$. The degeneration is smooth i.e., there are no conical singularities at $r = r_0$, provided we require that $ W(r)^{'}(r_0) \, b^2(r_0)^{'} = 1 $ or  
\begin{equation}
(1-\alpha^2)(2+\alpha^2)^2 = 1
\end{equation}
where, $\alpha = r_0/j$. This  cubic  has  a  unique  positive  solution  at $ \alpha^2 \approx 0.870385 $,  and in particular $r_0^2 < j^2$. With these relationships between the parameters, it can be checked that, $W(r),\, b(r)^2>0$ for $r>r_0$ and the spacetime metric is globally regular. Further
\begin{equation}
g^{tt} = - \frac{4 b(r)^2}{r^2 W(r)} < 0
\end{equation}  so the spacetime is stably causal, and in particular the $t=$constant hypersurfaces are Cauchy surfaces. Using the relationships between the parameters, it can also be checked that  
\begin{equation}
g_{tt} = -\frac{r^2 W(r)}{4b(r)^2} + b(r)^2 f(r)^2 <0 \quad \mbox{everywhere}.
\end{equation}
Hence, $\partial/ \partial t$ is globally timelike and there are no ergoregions. Hence the solutions to the wave equation do not suffer from Friedman's ergosphere instability recently proved in  \cite{Moschidis:2016zjy}.  In summary the above metric extends globally to a complete, asymptotically flat metric. Near $r = r_0$, the geometry of the manifold is that of $\mathbb{R} \times \mathbb{R}^2 \times S^2$ ($\partial_\psi$ degenerates at the origin of the $\mathbb{R}^2$ in the $(r,\psi)$ coordinates) and the $S^2$ has radius $r_0$ and is parameterized by $(\theta, \phi)$. 

The ADM mass and angular momenta of the soliton are
\begin{equation}
M = \frac{3\pi}{8} \left(\frac{r_0}{j}\right)^4 (j^2 + r_0^2), \qquad J_\psi = \frac{\pi r_0^6} {4 j^3}, \qquad J_\phi =0\, . 
\end{equation} In terms of angular momenta $(J_1, J_2)$ measured with respect to two orthogonal independent planes of rotation at infinity, this class of solitons is `self-dual' i.e., $J_1 = J_2$.  We note that more general solutions exist with $J_1 \neq J_2$, in which case the isometry group is broken to $\mathbb{R} \times U(1) \times U(1)$. 
Physically, the 2-cycle $[C]$ is prevented from collapse by a `dipole' flux
\begin{equation}
\mathcal{Q}: = \frac{1}{4\pi} \int_{S^2} F = \frac{\sqrt{3} r_0^2}{4j} \, ,
\end{equation} and these variables satisfy a `first law' of soliton mechanics $\td M = \Psi[C] \td \mathcal{Q}$ where $\Psi[C]$ is a certain intensive thermodynamical variable conjugate to $\mathcal{Q}$~\cite{Gunasekaran:2016nep}.

\subsection{Trapping of null geodesics}
\label{sec:nullgeo}
Let us now consider the properties of null geodesics in this spacetime. We will prove here that there is a region in the phase space of parameters for which null geodesics are stably trapped. A similar analysis was carried out for supersymmetric microstate goemetries in \cite{Eperon:2017bwq}. 

We start with the fact that the Hamilton-Jacobi function for null geodesics in \eqref{solmetric} is separable due to the existence of a reducible Killing tensor. In other words, the equations describing null geodesics are integrable. We write the Hamilton-Jacobi function $S$ in the separable form
\begin{equation}
S = -Et + R(r) + \Theta(\theta) + \psi p_\psi + \phi p_\phi
\label{hjeq}
\end{equation}
where $E,p_\psi \mbox{ and } p_\phi $ are conserved quantities associated to the three commuting Killing vector fields $ \partial_t,\, \partial_\psi,\, \partial_\phi $. We have another conserved quantity $C$ which is a separation constant arising from a reducible Killing tensor. Altogether, we have four constants of motion from the isometries of the solution. The conserved momenta can be obtained from the Hamiltonian $H = g^{ab} p_a p_b$ :
\begin{equation}
\begin{aligned}
& p_t  = -E = \left( \frac{-r^2 W(r)}{ 2 b(r)^2} + 2 b(r)^2 f(r)^2  \right) \dot{t} + 2 b(r)^2 f(r) \left( \dot{\psi} + \cos \theta \dot{\phi} \right) \,, \quad 
p_\psi  = 2 b(r)^2 (\dot{\psi} + f(r) \dot{t} + \cos \theta \dot{\phi}) 
\nonumber \\
& p_\phi  =\frac{r^2}{2} \sin^2 \theta \dot{\phi} + \cos \theta p_\psi \,, \quad
C = \left( \cot \theta p_\psi - \frac{1}{\sin \theta} p_\phi \right)^2 + \frac{ r^4 \dot{\theta}^2}{4}  \, .\label{constC}
\end{aligned}
\end{equation} The Hamilton-Jacobi function satisfies 
$$ \frac{\partial S}{\partial x^\mu }\frac{\partial S}{\partial x^\nu } g^{\mu \nu} =0 \, , $$
which gives
\begin{equation}
\label{nullgeo}
\frac{-4 b(r)^2}{r^2 W(r)} \left( E + f(r) p_\psi \right)^2 + R^{'}(r)^2 W(r) + \frac{4 C}{r^2} + \frac{p_\psi^2}{b(r)^2} = 0 \, .
\end{equation}
We can relate $R^{'}(r) $ and $ \Theta^{'}(\theta)$ to $\dot{r}$ and $ \dot{\theta}$ by,
\begin{equation}
\dot{x^{\mu}} = g^{\mu \nu} \frac{\partial S}{\partial x^{\nu} }
\end{equation}
which gives,
\begin{align}
R^{'}(r) = \frac{ \dot{r}}{2 W(r)} \,, \quad
\Theta^{'}(\theta) = \frac{r^2 \dot{\theta}}{8}.
\end{align}
This allows \eqref{nullgeo} to be rewritten as, 
\begin{align}
\frac{-4 b(r)^2}{r^2 W(r)} \left( E + f(r) p_\psi \right)^2 + \frac{ \dot{r}^2}{4 W(r)} + \frac{4 C}{r^2} + \frac{p_\psi^2}{b(r)^2} = 0  \, .
\end{align}
In summary, the equations for null geodesic $x^\alpha(\lambda)$ are given by 
\begin{align}
\dot{r}^2  &= -\frac{4 W(r) p_\psi^2}{ b(r)^2} + \frac{16 b(r)^2}{r^2} \left( E + f(r) p_\psi \right)^2  -\frac{16 W(r) C}{r^2}
\label{rdoteq} \\ 
\dot{\theta}^2 &= \frac{64}{r^4} \left[ C - \left( \cot \theta p_\psi - \csc \theta p_\phi \right)^2  \right] 
\label{thdot}\,,\,   \dot{t} = \frac{8 b(r)^2}{r^2 W(r)} \left( E + f(r) p_\psi \right) \\   
\dot{\phi} &= \frac{8 \csc \theta}{r^2}(\csc \theta p_\phi - \cot \theta \, p_\psi) \,, \quad
\dot{\psi} = -f(r)\dot{t} + 
\frac{2 p_\psi}{ b(r)^2} + \frac{8 \cot \theta}{r^2} \left( \cot \theta p_\psi - \csc \theta \, p_\phi \right) \, .
\end{align} 
From \eqref{rdoteq}, we can see that close to $r=r_0$ the first term dominates over the others  making $\dot{r}^2$ negative. This means null geodesics with non-zero $p_\psi$ approaching the `origin' must turn around at some $r> r_0$.  To simplify the analysis it is sufficient to restrict to motion in a plane with constant $\theta$.  Such null geodesics confined to a plane are solutions to $\ddot{\theta} =0$ with $\dot{\theta} =0$. For example, from the equation for $\dot{\theta}^2$ i.e., \eqref{thdot}, we can see that $C=0$ corresponds to geodesics confined in the $\theta=\pi/2$ equatorial plane. 


Stable trapping occurs when there is a region $[r_1,r_2]$ in which $ \dot{r}^2 >0 $ in the interior and vanishes at $r_2$ with $\dot{r}^2<0$ immediately outside the closed interval. Hence $r_1, r_2$ are turning points. Hence, stable trapping occurs when \eqref{rdoteq} has more than one turning point as depicted in Figs.\ref{stable1} and \ref{stable2}.


\begin{figure}[!htb]
	\centering
	\begin{minipage}{.34\textwidth}
		\centering
		\includegraphics[scale=0.39]{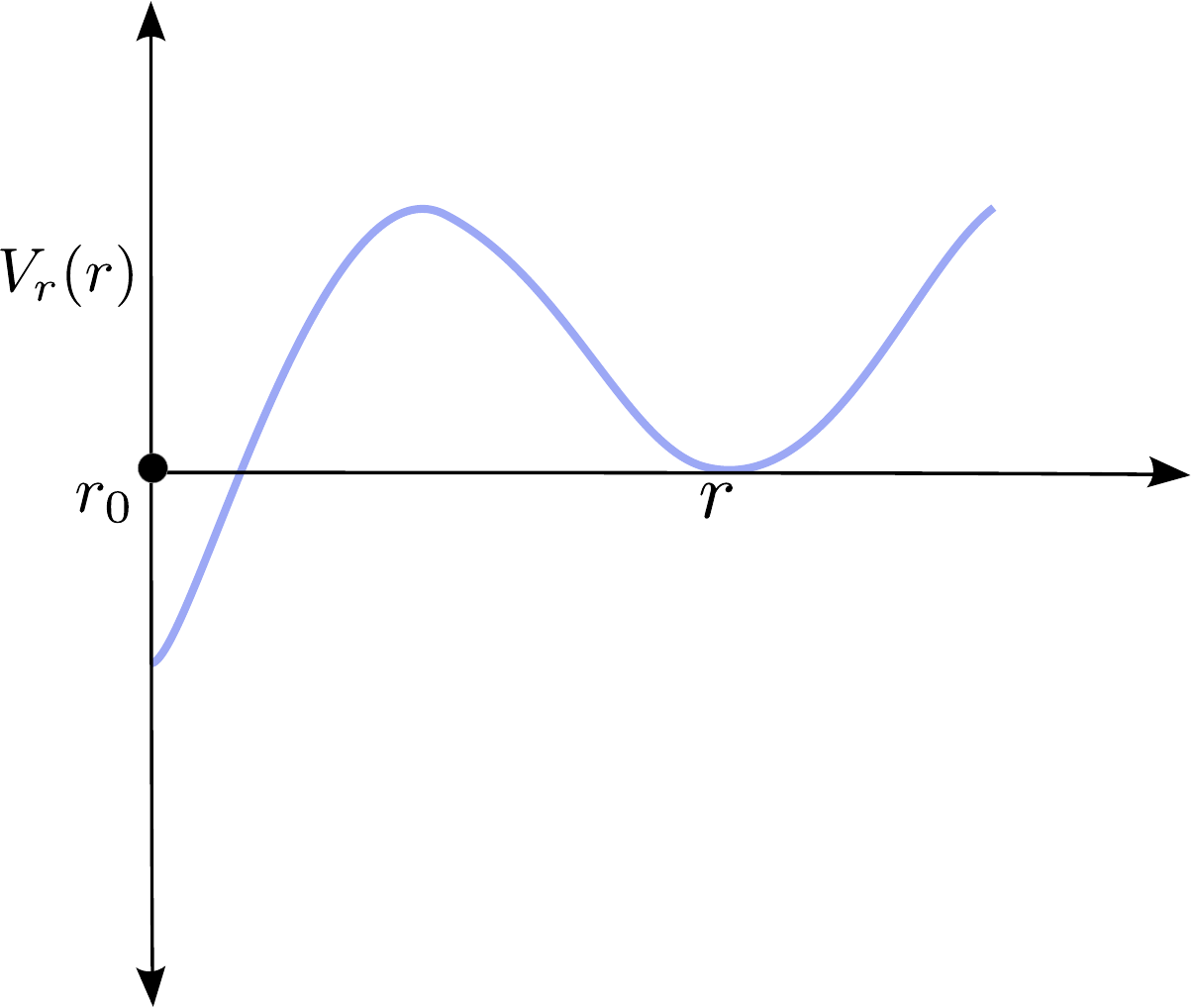}
		\caption{Unstable trapping}
		\label{unstable_trapping}
	\end{minipage}%
	\begin{minipage}{0.68\textwidth}
		\begin{figure}[H]
			\subfloat[ Case 1 (for $C=0$) \label{stable1} ]{\includegraphics[scale=0.4]{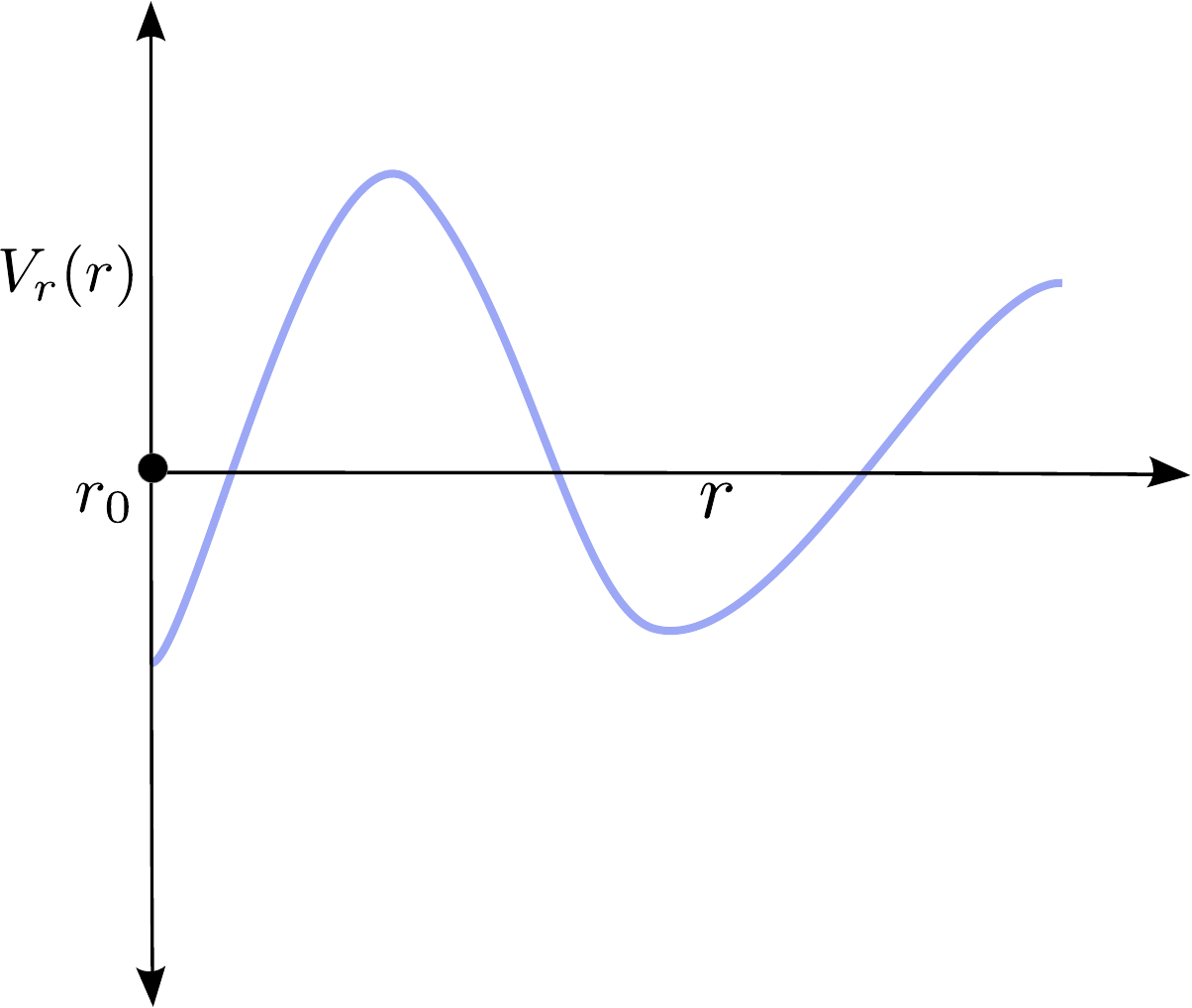}}
			\subfloat[Case 2 (for $C\neq0$) \label{stable2}]{\includegraphics[scale=0.4]{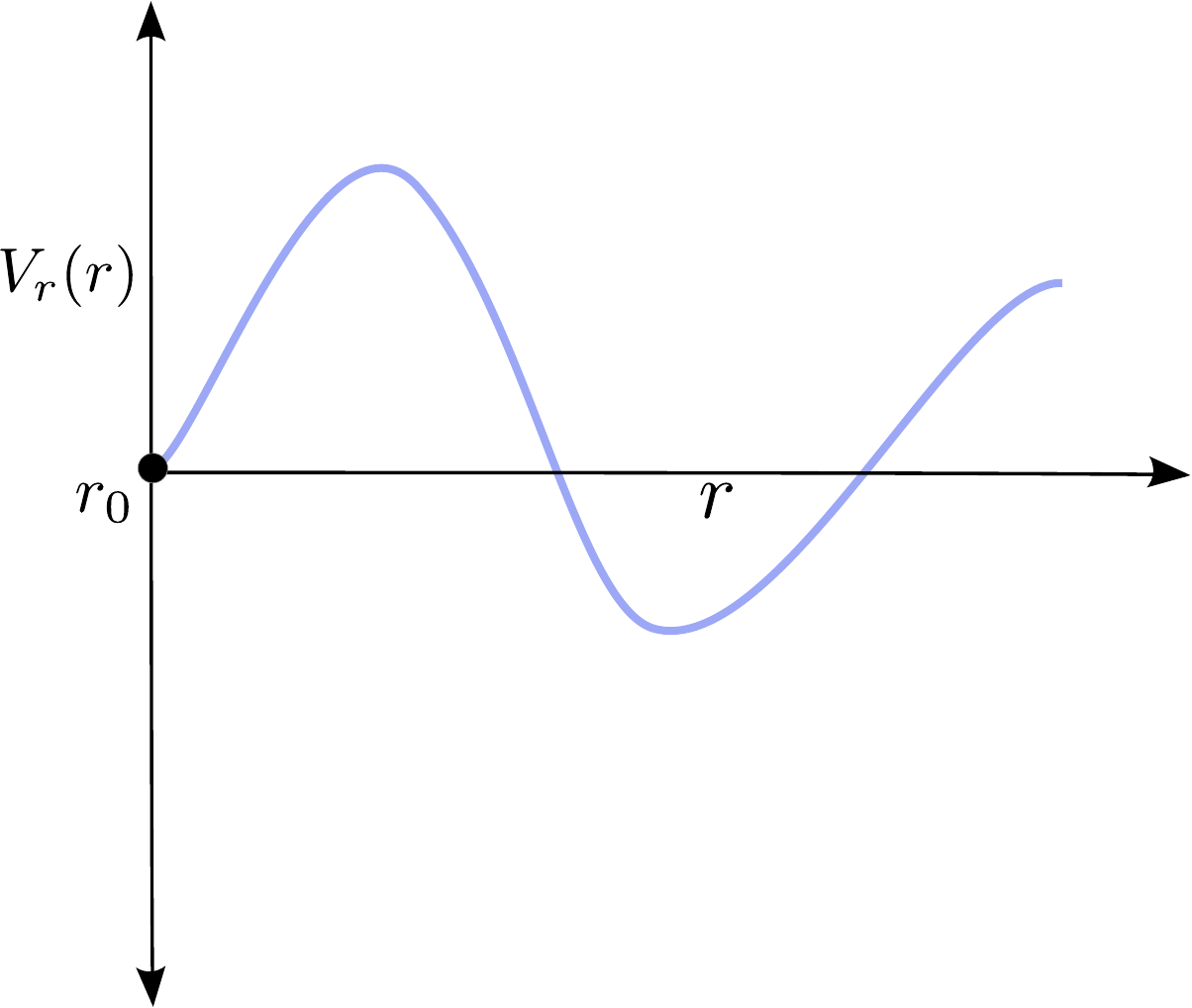}}
			\caption{Stable trapping}
		\end{figure}
	\end{minipage}
\end{figure}

\begin{claim}[Existence of stably trapped null geodesics]
	There exists a region in the phase space of parameters (of geodesic motion) for which the 1-parameter family of spacetimes given by \eqref{solmetric} exhibits stable trapping of null geodesics.  
\end{claim}

\begin{proof}
	There are two possible cases to consider :  $C=0$ and $C\neq 0$. We examine each of the cases below.
	
	\begin{enumerate}[label=(\alph*)]
		\item For $C=0$, we rewrite  \eqref{rdoteq} as 
		\begin{align}
		\dot{r}^2 = V_{r}(r) := -\frac{4 W(r) p_\psi^2}{ b(r)^2} + \frac{16 b(r)^2}{r^2} \left( E + f(r) p_\psi \right)^2.
		\label{rdotnull}
		\end{align}
		Stable trapping corresponds to $V_r(r)$ having at least two zeros. 
		It is useful to work with dimensionless quantities and scale out the dependence of $j$, so we use the following scaling for coordinates and parameters, 
		\begin{equation}
		r = j \cdot x ,\,\, \quad r_0 = \alpha \cdot j \,\, \quad \text{and} \quad E = \frac{\widetilde{E}}{j}
		\label{dimensionless}
		\end{equation}
		We only need positive roots greater than $\alpha^2$ to the equation. We recall that $\alpha^2 \approx 0.870385$. 
		With the following definitions, 
		\begin{align}
		\di \eta := \di \frac{p_\psi}{\widetilde{E}} \quad \mbox{ and } y: = x^2
		\label{scaleEx}
		\end{align} 
		\eqref{rdotnull} becomes, 
		\begin{align}
		-\frac{y^3 j^2}{4 \widetilde{E}^2} V_r(y) = -{y}^{3}+4\,{\eta}^{2}{y}^{2}+ \left( 4\,{\alpha}^{6}\eta-{\alpha}^{6}-4
		\,{\alpha}^{6}{\eta}^{2}+{\alpha}^{4}-4\,{\alpha}^{4}{\eta}^{2} \right) 
		{y}-4\,{\alpha}^{8}\eta+{\alpha}^{8}+4\,{\alpha}^{8}{\eta}^{2}
		\label{potequat}
		\underline{}\end{align}
		For $\eta \in (-1.33, -1.24)$, we have three turning points all bigger than $r_0$ indicating the existence of stably trapped null geodesics. This case is pictorially depicted in Fig.\ref{stable1}.
		
		\item For $C\neq 0$, fix $\theta = \pi/2$. From $\dot{\theta} =0$ we have $C=p_\phi^2$ and $\ddot{\theta} =0$ gives $p_\psi=0$. \eqref{rdotnull} becomes
		\begin{align}
		\dot{r}^2  = \frac{16 b(r)^2}{r^2}E ^2 - \frac{16 W(r) p_\phi^2}{r^2}
		\label{radp1}
		\end{align}
		Rewriting using \eqref{dimensionless} and \eqref{scaleEx}, \eqref{radp1} becomes, 
		\begin{align}
		\nonumber \frac{ j^4 y^3 V_r(y)}{\widetilde{E}^2} = \, & {y}^{3}-4\,{\eta}^{2}{y}^{2}+ \alpha^4\left( {\alpha}^{2}-1+4{\eta}
		^{2}{\alpha}^{2}+4{\eta}^{2} \right) y  + \alpha^4 \left( 4{\eta}^{2} - 4{\eta}^{2}{
			\alpha}^{2}-{\alpha}^{4}-4{\eta}^{2}{\alpha}^{4} \right)
		\end{align} 
		We know that $ V_r (\alpha^2) = 0$ and $ V_r(y) \rightarrow 1 $ as $ y \rightarrow \infty $. So, $V_r(y)$ is positive for large values of $y$. The two roots of $V(y)$ are
		\begin{align}
		y_1 &= -\frac{1}{2}\,{\alpha}^{2}+2\,{\eta}^{2}+\frac{1}{2}\,\sqrt {{\alpha}^{4}+8\,{\alpha}^{2}{\eta}^{2}+16\,{\eta}^{4}-4\,{\alpha}^{6}-16\,{\eta}^{2}{\alpha}^{6}-16 \,{\eta}^{2}{\alpha}^{4}} \\
		y_2 &= -\frac{1}{2}\,{\alpha}^{2}+2\,{\eta}^{2}-\frac{1}{2}\,\sqrt {{\alpha}^{4}+8\,{\alpha}^{2}{\eta}^{2}+16\,{\eta}^{4}-4\,{\alpha}^{6}-16\,{\eta}^{2}{\alpha}^{6}-16 \,{\eta}^{2}{\alpha}^{4}}
		\end{align}
		
		There is a double root (unstable trapping) if $ (\alpha^2 + 4 \eta^2)^2 - 4 \alpha^4 (\alpha^2 + 4 \eta^2 \alpha^2 + 4 \eta^2) = 0$. The real values of $\eta$ that solve this equation are
		\begin{align}
		\eta = 
		\pm \sqrt{{\frac {\alpha}{4} \left( 2\,{\alpha}^{2}+2\,\sqrt {{\alpha}^{3}
					\left( \alpha+2 \right) }+2\,\alpha-1 \right) }}
		\end{align}
		Hence there is a region of phase space corresponding to unstable trapping. With the solved value of $\eta^2$, we can also verify that $y_1 $ and $ y_2$
		(in this case $y_1=y_2$) are bigger than $\alpha^2$. This trapping structure is depicted in Fig.\ref{unstable_trapping}.
		For stable trapping, $V_r(y)$ should have three distinct positive roots. Clearly, $V_r(\alpha^2)=0$. We require that $y_2$ is real and $y_2-\alpha^2 >0$. This will hold provided $\eta$ satisfies
		\begin{equation}
		\eta^2 < \frac{{\alpha}^{4} + 2 \alpha^2}{8  -  4\,{\alpha}^{4} - 4 \,{\alpha}^{2}} \, .
		\end{equation}
		Also $y_1 > y_2 >  \alpha^2$ automatically. Hence there is a range of $\eta$ in phase space for which null geodesics are stably trapped. This is depicted in Fig.\ref{stable2}.
	\end{enumerate}
\end{proof} As discussed in the introduction, the above result suggests that waves with sufficiently high frequency will not decay rapidly enough to guarantee inverse polynomial decay for nonlinear applications. 

\section{Uniform boundedness}
\label{sec:boundedness} In this section we collect some basic results on solutions to the wave equation in this spacetime. Consider a solution $\Phi$ to the linear wave equation \eqref{lwave_eq}.  The energy momentum tensor associated with the field $\Phi$ is
\begin{align}\label{T_{ab}}
Q_{\mu \nu} & =  \nabla_\mu \Phi \nabla_\nu \Phi - \frac{1}{2} g_{\mu \nu} \nabla^\alpha \Phi \nabla_\alpha \Phi, 
\end{align}  which satisfies the conservation equation $\nabla^\mu Q_{\mu \nu} =0$. 
We introduce an orthonormal frame of one forms so that  the spacetime metric  \eqref{solmetric} can be expanded  as $g = \eta_{ab} \omega^a \omega^b$ where $\eta = \text{diag}(-1,1,1,1,1)$: 
\begin{equation}
\omega^0 = \frac{r \sqrt{W}}{2b} \td t, \qquad \omega^r = \frac{\td r}{\sqrt{W}}, \qquad \omega^1 = \frac{r}{2}\sigma^1, \qquad \omega^2 = \frac{r}{2} \sigma^2, \qquad \omega^3 = b (\sigma^3 + f \td t) \, .
\end{equation} The dual orthonormal frame of vector fields satisfying $g^{-1} = \eta^{ab} e_a e_b$ is
\begin{equation}
e_0 = \frac{2b}{r\sqrt{W}} (\partial_t - f L_3), \qquad e_r = \sqrt{W} \partial_r, \qquad e_1 = \frac{2}{r} L_1, \qquad e_2 = \frac{2}{r} L_2, \qquad e_3 = \frac{L_3}{b}
\end{equation} where $L_i$ are the vector fields dual to the left-invariant one-forms $\sigma^i$, i.e. $\sigma^i (L_j) = \delta^i_j$, $i = 1,2,3$.  The unit normal to a $t=$ constant surface is $n = -\omega^0$. As a vector field the unit future-pointing normal is $N = e_0$.  Note that $n \propto -\td t$.   The timelike Killing vector field $T = \partial_t$ is, in this frame, 
\begin{equation}
T = \frac{r\sqrt{W}}{2b} e_0 + f b e_3 \, . 
\end{equation} The current $J^T[\Phi]_a = Q_{ab}T^b$ associated to this vector field is
\begin{equation}
J^T[\Phi] =\left (\frac{r\sqrt{W}}{2b} e_0 (\Phi)  + f b e_3(\Phi)\right)d\Phi + \frac{1}{2}\left( \frac{r \sqrt{W}}{2b} \omega^0 - f b \omega^3 \right) |\td \Phi|^2
\end{equation} where 
\begin{equation}
|\td \Phi|^2 = -(e_0 (\Phi))^2 + \sum_{i=1}^4 (e_i (\Phi))^2
\end{equation}  Since $T, N$ are future directed, timelike vector fields,  the scalar $Q(T, N)$ must be positive definite.  We can observe this quite explicitly by computing
\begin{equation}
J^T(N)[\Phi] := Q(T,N) = \frac{r \sqrt{W}}{4b} (e_0(\Phi))^2 + f b e_3(\Phi) e_0(\Phi) + \frac{r \sqrt{W}}{4b} \sum_{i=1}^4 (e_i(\Phi))^2
\end{equation} and then using Young's inequality
\begin{equation}
J^T(N)[\Phi] \geq   \frac{r \sqrt{W}}{4b} (e_0(\Phi))^2 + \frac{r \sqrt{W}}{4b} \sum_{i=1}^4 (e_i(\Phi))^2 - \frac{f b}{2} (e_0(\Phi))^2 - \frac{ fb}{2} (e_3(\Phi))^2 \geq C \sum_{\alpha=0}^4 (e_\alpha (\Phi))^2
\end{equation} where we have noted that $g_{tt} < 0$ implies that 
\begin{equation}
\frac{ r\sqrt{W}}{ 2 b} >| f b|.
\end{equation}  Let $\Sigma_t$ denote a spatial hypersurfface defined by $t =$ constant with induced metric $h$. From the above,  the following first-order energy associated to $\Sigma_t$ is non-negative: 
\begin{equation}
\label{energy}
E[\Phi](t) : = \int_{\Sigma_t} J^T(N)[\Phi] \; \td\text{Vol}_h \sim \int_{\Sigma_t}  \sum_{\alpha=0}^4 (e_\alpha(\Phi))^2 \; \td \text{Vol}_h
\end{equation} and in the following we show that it is controlled by the energy of the initial data. We will use the symbol $E_\Omega[\Phi](t)$ to represent the same integral as above with the region of integration replaced with $\Omega \cap \Sigma_t $ where $\Omega$ is a spacetime region. 
If $T$ is a timelike Killing vector field, one finds that the current is conserved. From the divergenceless of $Q_{\mu \nu}$ it is easy to see
\begin{equation}
\nabla^\mu J^T_\mu (\Phi) = 0. 
\end{equation}
Let $\Sigma_0 $ and $\Sigma_t$ be two homologous surfaces with a common boundary. 
Integrating $J^T_\mu(\Phi)$ over the region enclosed by $\Sigma_0 $ and $\Sigma_t$, whose normals are $n^\mu_0$ and $n^\mu_t$ respectively, and using the divergence theorem, we get, 
\begin{align}
\int_{\Sigma_t} J^T_\mu(\Phi) n_t^\mu = \int_{\Sigma_0} J^T_\mu(\Phi) n_0^\mu  
\end{align}
This holds as long as $T$ is timelike. For the soliton spacetime \eqref{solmetric}, we have a global timelike Killing vector field. No part of $\Sigma_t \mbox{ or } \Sigma_0$ is null and hence the control on $\Phi$ and its derivatives does not degenerate anywhere. We thus quite straightforwardly obtain the following uniform energy bound.
\begin{equation}
E[\Phi](t) = E[\Phi](0).
\end{equation} Finally, we define higher-order energies
\begin{equation}
\label{higherenergy}
E_k[\Phi](t) : = \sum_{0 \leq |\alpha| \leq k -1 }\int_{\Sigma_t} J^T(N)[\partial_\alpha \Phi] \; \td \text{Vol}_h =  \sum_{0 \leq |\alpha| \leq k -1 } E_t[\partial_\alpha \Phi]
\end{equation} These energies are roughly equivalent to the sum of the homogeneous seminorms $\dot{H}^k$ on $\Sigma_t$ with $s \in [1,k]$.  

\section{Separation of variables and eigenvalue problems}
\label{sepvar}

\subsection{Separation of variables} A preliminary step towards the construction of quasimodes is the separation of variables of the wave equation to reduce the problem  to a one-dimensional Schr\"{o}dinger type equation. The advantage to the class of geometries we are considering is that, due to the $\mathbb{R} \times SU(2) \times U(1)$ isometries, apart from a single radial equation, the remaining parts of the wave equation can be solved explicitly, and in particular the spectrum is completely understood. This simplification also allows us to observe how trapping manifests at the wave equation level by studying an effective potential in the radial equation. In the metric given by \eqref{solmetric} we set 
\begin{equation}\label{psidef}
\di \widetilde{\psi} = \di \frac{\psi}{2} \implies \frac{\partial}{\partial \psi} = \frac{1}{2} \frac{\partial }{\partial \widetilde{\psi}}\,, \qquad \hat{b}(r)^2 = 4 b(r)^2, 
\end{equation} so that $\widetilde{\psi} \sim \di \widetilde{\psi} + 2 \pi$. This normalization is consistent with the conventions used in \cite{kunduri:2006qa}.  We can rewrite \eqref{solmetric} as
\begin{align}
\td s^2 = \frac{-r^2 W(r)}{  \hat{b}(r)^2} \td t^2 + \frac{\td r^2}{W(r)} + \frac{r^2}{4} ( \sigma_1^2 + \sigma_2^2 ) + \hat{b}(r)^2 \left( \td \widetilde{\psi} + \frac{\cos \theta}{2} \td \phi +  f(r) \td t \right)^2
\label{solmetric1}
\end{align}
For later reference we record the inverse metric: 
\begin{align}
\left( \frac{\partial}{\partial s} \right)^2  = & -\frac{\hat{b}(r)^2}{r^2 W(r)}  \left( \frac{\partial }{\partial t} - \frac{f(r)}{2} \frac{\partial}{\partial \widetilde{\psi} } \right)^2 + W(r) \left( \frac{\partial}{\partial r} \right)^2 + \frac{4}{r^2} \left( \frac{\partial }{\partial \theta} \right)^2  \nonumber \\ & + \frac{4}{r^2} \left( \frac{\cot \theta}{2} \frac{\partial}{\partial \widetilde{\psi}} - \frac{1}{\sin \theta} \frac{\partial}{\partial \phi}  \right)^2 + \frac{1}{\hat{b}(r)^2} \left( \frac{\partial}{\partial \widetilde{\psi}} \right)^2 \,  ,  
\end{align} and volume form  is given by,
\begin{align}
\td \textrm{Vol}_g = \frac{r^3}{4} \sin \theta \td t \wedge \td r \wedge \td \widetilde{\psi} \wedge \td \theta \wedge \td \phi
\end{align}
The wave equation can be explicitly written out as,
\begin{align}
\Box_g \Phi = 
\frac{1}{r^3} \frac{\partial}{\partial r} \left( r^3 W(r)  \frac{\partial \Phi}{\partial r}  \right) + \frac{4}{r^2 \sin \theta} \frac{\partial }{\partial \theta} \left( \sin \theta  \frac{\partial \Phi}{\partial \theta} \right) + g^{AB} \frac{\partial^2 \Phi }{\partial x^A \partial x^B} 
\end{align} where $A,B = t,\phi,\widetilde{\psi}$ run over the ignorable coordinates and
\begin{align}
\nonumber g^{AB} \frac{\partial^2 }{\partial x^A \partial x^B} = - \frac{\hat{b}(r)^2}{r^2 W(r)} \left( \frac{\partial }{\partial t} - \frac{f(r)}{2} \frac{\partial }{\partial \widetilde{\psi}} \right)^2 + \frac{4}{r^2} \left( \frac{\cot \theta}{2} \frac{\partial }{\partial \widetilde{\psi}} - \frac{1}{\sin \theta} \frac{\partial}{\partial \phi} \right)^2  + \frac{1}{\hat{b}(r)^2} \frac{\partial^2 }{\partial \widetilde{\psi}^2}
\end{align}
The isometry group suggests we seek separable solutions of the form 
\begin{equation}
\Phi(t,r,\theta,\widetilde{\psi},\phi) = e^{-i \hat{\omega} t} e^{i n \widetilde{\psi}} R(r) Y(\theta,\phi) . 
\label{antsatz}
\end{equation}
With the separation ansatz \eqref{antsatz},
\begin{align}
g^{AB} \frac{\partial^2 }{\partial x^A \partial x^B} = \frac{\hat{b}}{r^2 W} \left( \hat{\omega} + \frac{fn}{2}\right)^2 \Phi - \frac{n^2}{\hat{b}^2} \Phi + \frac{4}{r^2} \left( \frac{\cot \theta}{2} \frac{\partial }{\partial \widetilde{\psi}} - \frac{1}{\sin \theta} \frac{\partial}{\partial \phi} \right)^2 \Phi
\end{align}
Consider the following round metric on $S^2$:
\begin{align}
\hat{g}_{\ti \tj} \td x^{\ti} \td x^{\tj} = \frac{1}{4} \left( \td \theta^2 + \sin^2 \theta  \td \phi^2 \right)
\end{align} 
normalized so that $ \mbox{Ric}(\hat{g}) = 4 \hat{g}$. Define the 1-form $A = \di \frac{\cos \theta}{2} \td \phi$ which is locally defined on $S^2$ which is easily seen to be a potential for the K\"ahler form on $\mathbb{CP}^1 \cong S^2$. Let
$$ D  := \nabla_{S^2} - i n A $$
We have
\begin{equation}
\begin{aligned}
D^2 & = \hat{g}^{\ti \tj} D_{\ti} D_{\tj} = g^{\ti \tj}\left( (\nabla_{S^2})_{\ti} - i n A_{\ti} \right) \left((\nabla_{S^2})_{\tj} - i n A_{\tj} \right) \nonumber \\ &= \Delta_{S^2} - 2 i n A_{\ti} \hat{g}^{\ti \tj} (\nabla_{S^2})_{\tj} - in \mbox{div}_{\hat{g}} A - n^2\hat{g}^{\ti \tj}
A_{\ti} A_{\tj} 
\end{aligned}
\end{equation} Since $\mbox{div}_{\hat{g}} A = \hat{g}^{\ti \tj} (\nabla_{S^2})_{\ti} A_{\tj} = 0 $,
\begin{align}
D^2 &= \Delta_{S^2} - 2 i n A_{\ti} \hat{g}^{\ti \tj} (\nabla_{S^2})_{\tj} - n^2\hat{g}^{\ti \tj}
A_{\ti} A_{\tj} 
\end{align}
We now compute $D^2$ explicitly. The Laplacian on $S^2$ is
\begin{align}
\Delta_{S^2} & = \frac{4}{\sin \theta} \partial_\theta \left(\sin \theta \partial_\theta  \right) + \frac{4}{\sin^2 \theta} \partial^2_\phi \end{align}
and the remaining terms are
\begin{align}
- 2 i n A_{\ti} \hat{g}^{\ti \tj} (\nabla_{S^2})_{\tj}  &= -4 in \frac{\cos \theta}{\sin^2 \theta} \partial_\phi \nonumber \\ 
n^2\hat{g}^{\ti \tj} A_{\ti} A_{\tj} &= n^2 \frac{4}{\sin^2 \theta} \frac{\cos^2 \theta}{4} = n^2 \cot^2 \theta 
\end{align} which gives
\begin{align}
D^2 = \frac{4}{\sin \theta} \partial_\theta (\sin \theta \partial_ \theta) + \frac{4}{\sin^2 \theta} \partial_\phi^2 - n^2 \cot^2 \theta - 4 in \frac{\cos \theta}{\sin^2 \theta} \partial_\phi
\end{align} The operator $D^2$ is the charged Laplacian on $S^2$ and its spectrum has been analyzed in detail in the context of $U(1)$ monopoles.  Its eigenfunctions $Y(\theta, \phi)$ (suppressing the eigenvalue labels) are similar to the standard spherical harmonics.
\begin{equation}
D^{2} Y(\theta,\phi) = -\mu Y(\theta,\phi) 
\end{equation}\label{spectrum}
$\mu \geq 0$ are a discrete family of eigenvalues with corresponding eigenfunctions $Y(\theta,\phi)$ \cite{Wu:1976ge,kunduri:2006qa}. The values taken by $\mu$ are,
\begin{align}
\mu = \ell(\ell+2) - n^2\,, \mbox{ where } \, \ell = 2 K + |n| \,\, \mbox{ with } K=0,1,2,3...
\end{align} For simplicity throughout this work we will suppress the eigenvalue labels that characterize the eigenfunctions; generally we will work with individual modes with eigenvalue $\mu$. 
We can concisely write the wave operator on $g$ as 
\begin{align}
\Box_g \Phi = \frac{1}{r^3} \partial_r \left( r^3 W \partial_r \Phi \right) + \frac{D^2}{r^2} \Phi + \frac{\hat{b}}{r^2 W} \left( \hat{\omega} + \frac{nf}{2} \right)^2 \Phi - \frac{n^2}{\hat{b}^2} \Phi
\end{align}
The wave equation with the separation ansatz \eqref{antsatz} becomes
\begin{align}
&\frac{1}{r^3} e^{-i \hat{\omega} t} e^{i n \widetilde{\psi}} Y(\theta,\phi) \frac{\td R(r)}{\td r} \left( r^3 W(r) \frac{\td R}{\td r} \right) + e^{-i \hat{\omega} t} e^{i n \widetilde{\psi}} \frac{R(r)}{r^2} D^2 Y(\theta,\phi) \nonumber \\ & \quad + e^{-i \hat{\omega} t} e^{i n \widetilde{\psi}} \frac{\hat{b}(r)^2}{r^2 W(r)} \left( \hat{\omega} + \frac{n f(r)}{2}  \right)^2  R(r) Y(\theta,\phi) - \frac{n^2}{\hat{b}(r)^2} e^{-i \hat{\omega} t} e^{i n \widetilde{\psi}} R(r) Y(\theta,\phi) = 0
\end{align}
which finally reduces to,
\begin{align}
\label{rad_eq}
\frac{1}{r^3} \frac{\td}{\td r} \left( r^3 W(r) \frac{\td  R(r)}{\td r } \right) + \left\{ -\frac{\mu}{r^2} + \frac{\hat{b}(r)^2}{r^2 W(r)} \left( \hat{\omega} + \frac{n f(r)}{2} \right)^2 - \frac{n^2}{\hat{b}(r)^2}  \right\} R(r) =0
\end{align}
To recast this into a Sch\"odinger-like form, we make the following transformations. 
\begin{align}
R(r) = \frac{u}{r \sqrt{\hat{b}(r)} },\, \quad w = \int_{r_0}^{w} \frac{\hat{b}(s)}{s W(s)} \td s 
\label{rcoordtrans}
\end{align}
after which \eqref{rad_eq} becomes,
\begin{align}
\frac{\td}{\td r} \left( r^3 W(r) \frac{\td  R(r)}{\td r } \right) = \frac{\hat{b}(r)^{3/2}}{W(r)} \frac{\td^2 u}{\td w^2} - \left\{ \frac{r W(r)}{\sqrt{\hat{b}(r)}} + \frac{1}{2} \frac{r^2 W(r)}{\hat{b}(r)^{3/2 }} \frac{\td \hat{b}(r)}{\td r} \right\} u
\end{align}
which can be rewritten as,
\begin{align} \nonumber
-\frac{\td^2 u}{\td w^2} + \left\{ \frac{W(r)}{\hat{b}(r)^{3/2}} \left( \frac{r W(r)}{\sqrt{\hat{b}(r)}} \right.  \left. + \frac{1}{2} \frac{r^2 W(r)}{\hat{b}(r)^{3/2 }} \frac{\td \hat{b}(r)}{\td r} \right)     + \frac{W(r)}{\hat{b}(r)} \left(\mu - \frac{\hat{b}(r)^2}{ W(r)} \left( \hat{\omega} + \frac{n f(r)}{2} \right)^2 - \frac{n^2 r^2}{\hat{b}(r)^2}  \right) \right\} u = 0
\label{shrform}
\end{align}
Comparing with a Schr\"{o}dinger equation of the form
\begin{equation}
- \frac{\td^2 u}{ \td w^2 } + \widetilde{V} u = 0
\label{ode_main}
\end{equation}
we can read off the potential as $\widetilde{V}$,
\begin{align}
\widetilde{V} = \frac{W}{{\hat{b}}^{\frac{3}{2}}} \partial_r \left[ \frac{rW}{ {\hat{b}}^{\frac{1}{2}}} + \frac{1}{2} \frac{r^2 W \partial_r {\hat{b}}}{ {\hat{b}}^{\frac{3}{2}} } + \right] + \frac{W}{{\hat{b}}^2} \left[ \mu + \frac{n^2 r^2}{{\hat{b}}^2} - \frac{{\hat{b}}^2}{W} \left( \hat{\omega} + \frac{nf}{2}\right)^2  \right]
\label{potential}
\end{align} In summary  we have shown that not only can the wave equation be separated, but we can obtain explicit, analytic solutions for the separated solution apart from a single radial Schr\"{o}dinger equation.  This is in contrast with other stationary non-static solutions for which the angular part of the wave equation cannot be solved explicitly (e.g. Kerr or generic Myers-Perry black holes).  This nice property characteristic of cohomogeneity-one rotating black holes has been used in the study of linearized gravitational perturbations (see, e.g. \cite{kunduri:2006qa})

\subsection{Trapping of high frequency waves}
Consider solutions to  the wave equation 
$\Box_g \Phi =0$ which are of the form, $\Phi(\mathbf{y},t) = e^{- i  \hat{\omega} t} U(\mathbf{y}) $ where $\mathbf{y}$ refers to a spatial variable. Let $U(\mathbf{y})$ solve the following Schr\"{o}dinger type equation, 
\begin{align}
- \frac{\td^2 U}{ \td \mathbf{y}^2 } + (\mathbf{V} - \hat{\omega}^2) U = 0
\end{align}
where $\mathbf{V} := \mathbf{V}(\mathbf{y})$. Consider a structure of the potential $\mathbf{V}$ as depicted in Fig.\ref{potential_soliton}.
\begin{figure}[H]
	\includegraphics[scale=0.7]{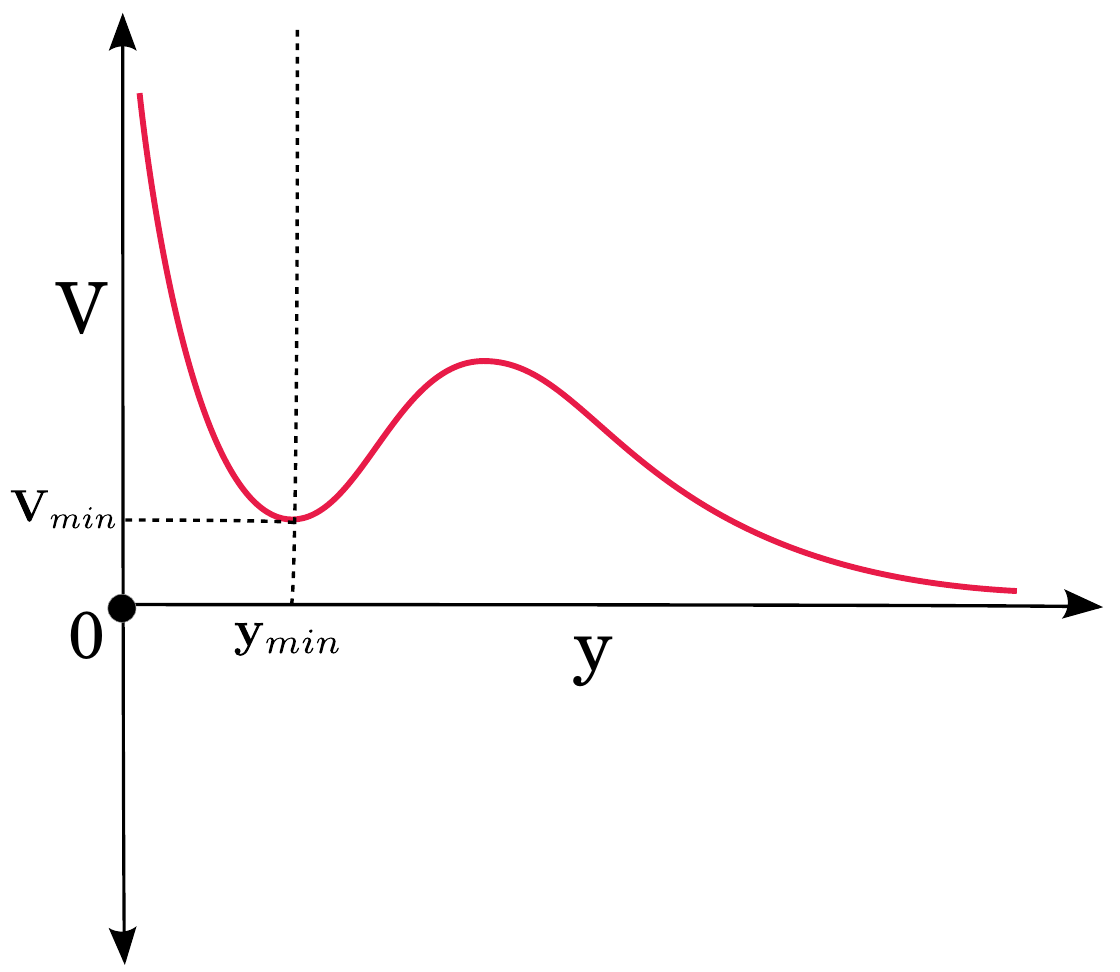}
	\caption{Structure of potential in stable trapping}
	\label{potential_soliton}
\end{figure}
A minimum in the potential $\mathbf{V}_{min}$ indicates that high frequency waves with \emph{suitable energy} ($\omega^2$) which roughly travel along null geodesics remain localized in the the region about $\mathbf{y}_{min}$. In other words, we say that high frequency waves are \emph{trapped}. One can intuitively see that this trapping ultimately leads to a slow decay of waves. The purpose of this work is  to prove this rigorously. 

In comparison to the discussion above, $ \widetilde{V} $ has a term dependent on $\hat{\omega}$ viz., $-nf\hat{\omega}$. Here, we will understand how to analyze the structure of $\widetilde{V}$. The expression for  $\widetilde{V}$, \eqref{potential} has two kinds of terms involving $\hat{\omega}$,  namely (a) $\hat{\omega}^2$ which is the eigenvalue and (b) $-nf\hat{\omega}$ which is a nonlinear term in the potential $\widetilde{V}$. We define $\widehat{V}$ as the part of the potential independent of $\hat{\omega}$ which is the analogue of $\mathbf{V}$ above. 
\begin{align}
\widehat{V} = \widetilde{V} + \hat{\omega}^2 + nf\hat{\omega} = & \frac{W^2}{{\hat{b}}^2} + \frac{rW (\partial_r W)}{{\hat{b}}^2} - \frac{r W^2 }{2} {\left(\frac{\partial_r{\hat{b}}}{{\hat{b}}^3}\right)} + \frac{r^2 W(\partial_r W)}{2} {\left(\frac{\partial_r{\hat{b}}}{{\hat{b}}^3}\right)}  + r W^2 \left(\frac{\partial_r {\hat{b}}}{{\hat{b}}^3}\right) \nonumber \\  & + \frac{W^2 r^2}{2} \left(\frac{{\partial^2_r} {\hat{b}}}{{\hat{b}}^3}\right) - \frac{3 W^2 r^2 {\hat{b}}^2}{4} \left(\frac{\partial_r {\hat{b}}}{{\hat{b}}^3}\right)^2  + \frac{W}{{\hat{b}}^2} \left[ \mu + \frac{n^2 r^2}{{\hat{b}}^2} -  \left(  \frac{nf}{2}\right)^2  \right] 
\end{align}
We use the following definitions to simplify the expressions:
\begin{align}
Y_1(r) &= \frac{\partial_r {\hat{b}} }{{\hat{b}}^3} ,\, \quad Y_2(r) = \frac{\partial^2_r b }{{\hat{b}}^3} = \partial_r(Y_1(r)) - 3{\hat{b}}^2 (Y_1(r))^2   .
\end{align}
Rewriting the potential in terms of $Y_1(r)$ and $Y_2(r)$ (chiefly to avoid the appearance of odd powers of $\hat{b}(r)$), we get the following expression, 
\begin{align}
\widehat{V} = & \frac{W^2}{{\hat{b}}^2} + \frac{rW (\partial_r W)}{{\hat{b}}^2} - \frac{r W^2 }{2 } Y_1(r) + \frac{r^2 W(\partial_r W)}{2} Y_1(r)  + r W^2 Y_1(r)  \nonumber \\  & + \frac{W^2 r^2}{2} Y_2(r) - \frac{3 W^2 r^2 {\hat{b}}^2}{4} Y_1(r)^2    \frac{\mu W}{{\hat{b}}^2}  + \frac{W n^2 r^2}{{\hat{b}}^4} -  \frac{n^2f^2}{4} .
\end{align}
$\widehat{V}$ has terms which depend on $n$ and $\mu$ and terms indepenent of these charged Laplacian eigenvalues. Hence, we decompose $\widehat{V} = \widehat{V}_{dom} + \widehat{V}_{j}$ where $\widehat{V}_{dom} $  is the dominant part of the potential. This reflects the fact that for large $n$ or $\mu$, $\widehat{V}_{dom}$ would be the term dictating the behaviour of the potential i.e., for sufficiently large $n$ and $\mu$, $ \widehat{V} \lesssim \widehat{V}_{dom}$.  
\begin{align}
\widehat{V}_{dom} &=   \frac{\mu W}{{\hat{b}}^2}  + \frac{W n^2 r^2}{{\hat{b}}^4} -  \frac{n^2f^2}{4}  \\ 
\widehat{V}_{j} = \widehat{V} - \widehat{V}_{dom} &=  \frac{W^2}{{\hat{b}}^2} + \frac{rW (\partial_r W)}{{\hat{b}}^2} - \frac{r W^2 }{2 } Y_1(r) + \frac{r^2 W(\partial_r W)}{2} Y_1(r) \nonumber \\  & + r W^2 Y_1(r) + \frac{W^2 r^2}{2} Y_2(r) - \frac{3 W^2 r^2 {\hat{b}}^2}{4} Y_1(r)^2 
\end{align}
We recall that the eigenvalues $n$ and $\mu $ are related as in~\eqref{spectrum} and 
$K$ can be independently chosen and here we choose it to vary as $n$. With this, the dependence of $\mu$ on $n$ is :
\begin{align}
\mu = 8n^2 + 6n
\end{align}
We can see that only the terms proportional to $n^2$ in $\widehat{V}_{dom}$  matter when $n$ is large. This happens to be the regime of $n$ we are interested in for reasons which will be given in the next section. $V_{dom}$ can be split up as, 
\begin{align}
\widehat{V}_{dom} = n^2 \widehat{V}_{\sigma_1} + n \widehat{V}_{\sigma_2} 
\end{align}
Explicitly, 
\begin{align}
\widehat{V}_{\sigma_1} = \frac{8 W}{{\hat{b}}^2} + \frac{W r^2}{{\hat{b}}^4} - \frac{f^2}{4} \quad \mbox{and} \quad \widehat{V}_{\sigma_2} = \frac{6W}{{\hat{b}}^2}
\end{align}
Hence, the ODE of interest is
\begin{align}
- \frac{\td^2 u}{ \td w^2 } + (n^2 \widehat{V}_{\sigma_1} - nf \hat{\omega} - \hat{\omega}^2) u = 0
\label{main_eq}
\end{align}
The results for \eqref{main_eq} will carry over for $(\widehat{V}_{\sigma_1} - nf \hat{\omega} - \hat{\omega}^2)$ replaced by the effective potential $\widetilde{V}$.
As described in Sec.\ref{minsugra}, the solution is a 1-paramter family. Here, we rewrite the differential equation in terms of dimensionless variables as we did while analyzing trapping of null geodesics. With the following rescalings, $ 
w = j x \,, r_0 = \alpha j \,\, \text{ and }  \di \omega = \hat{\omega}/j  $, (noting that $w$ scales the same way as $r$) \eqref{main_eq} becomes,
\begin{align}
\begin{aligned}
- \frac{1}{j^2}\frac{\td^2 u}{ \td x^2 } + \frac{1}{j^2}\left( n^2 V_{\sigma_1} - n \widetilde{f} \omega - \omega^2 \right) u &= 0 \\ \implies - \frac{\td^2 u}{ \td x^2 } + \left( V_{\sigma_1} - n \widetilde{f} \omega - \omega^2 \right) u &= 0  \label{main_eq1}
\end{aligned}
\end{align}
where, $V_{\sigma_1} = j^2 \widehat{V}_{\sigma_1}  \mbox{ and } \tilde{f} = j f$. $V_{\sigma_1}$ is explicitly given below,
\begin{align}
\begin{aligned}
V_{\sigma_1} = \frac{\left( {x}^{2} -{\alpha}^{2} \right)^{-1}}{16 \left( {\alpha}^{6}+{\alpha}^{2}{x}^{2}+{x}^{4} \right) ^{2}
	 } \di
\left\{ \right. & \left. 129 {\alpha}^{14} - \left( 129 {x}^{2}-128 \right) {\alpha}^
		{12}-128 {\alpha}^{10}+128\,{\alpha}^{8}{x}^{2} \right. \\ & \left. - \left(144\,{x}^{6}
		-128\,{x}^{4}+128\,{x}^{2} \right) {\alpha}^{6}- \left( 144 {x}^{6}+
		128\,{x}^{4} \right) {\alpha}^{4} \right. \\ & \left. +144 {x}^{6}{\alpha}^{2}+144 {x}^{8
	} \right\} 
\end{aligned}
\end{align}

As the first step, we confirm that the spacetime exhibits the structure for stable trapaping with a plot of $V_{\sigma_1}$ in Fig.\ref{sol_potential}. 
The minimum characterizes the stably trapped region and the region in the neighbourhood of the minimum, which is devoid of any local maxima will be denoted by $[x_{-},x_{+}]$. 
\begin{figure}[htb!]
	\includegraphics[scale=2.7]{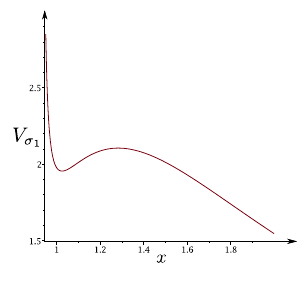}
	\caption{Plot of $ V_{\sigma_1}$ against $x $}
	\label{sol_potential}
\end{figure}
The aim of studying \eqref{main_eq1} is the construction of eigenfunctions in $[x_{-},x_{+}]$ with Dirichlet conditions which will be seen in subsequent sections. To give more relevance to this construction here, we give an informal introduction to quasimodes (see \cite{benomio:2018ivy}).

Consider time periodic functions of the form $\Psi_n(t,x) = e^{-i \omega_n t} u_n(x) $ where $\omega_n$ are real.  Quasimodes are approximate solutions of the form $\Psi_n$ to the wave equation satisfying the following properties:
\begin{enumerate}
	\item $\Psi_n$ belongs to an appropriate energy space. 
	\item They are localized in frequency and space i.e.,
	\begin{equation}
	\nonumber \norm{\partial^2 \Psi_n} \approx \omega_n^2 \norm{\Psi_n}
	\end{equation}
	and $\Psi_n$ are compactly supported. 
	\item They are an approximate solution to the wave equation i.e.,
	\begin{equation}
	\nonumber \Box_g \Psi_n = F_n(\Psi_n)
	\label{inhom}
	\end{equation}
	where $F_n(\Psi_n) \to 0$ as $n \to \infty$. Intuitively,  the error can be made small in an appropriate limit. 
\end{enumerate} In particular, consider the case where $ F_n(\Psi_n) \sim e^{-C n} $ where $C$ is any constant. By constructing an appropriate sequence of the approximate solutions  $\Psi_n$ one can establish that there are slow decaying solutions to the wave equation which contradicts any uniform fast decay statement.    

\subsection{Linear and nonlinear eigenvalue problems} The main eigenvalue problem that we study is the 
Schr\"{o}dinger type wave equation along with Dirichlet boundary conditions at $x_{-}$ and $x_{+}$. The problem is stated below
\begin{align}
\begin{aligned} 
- \frac{\td^2 u}{ \td x^2 } + & \left( n^2 V_{\sigma_1} - n \widetilde{f} \omega - \omega^2  \right) u = 0  \\ & u(x_{-}) = u(x_{+}) = 0  \label{shreq}
\end{aligned}
\end{align}
As mentioned previously the ``potential term" appearing here has a nonlinear dependence on $\omega$, which constitutes a nonlinear eigenvalue problem. This makes a straightforward analysis of \eqref{shreq} difficult. We define ${\mathcal{P}}_\beta$ to be the following family of eigenvalue problems labeled by $\beta \in [0,1]$. 
\begin{align} 
\label{evfamily}
{\mathcal{P}}_\beta \quad  : \begin{aligned} 
- \frac{\td^2 u}{ \td x^2 } + & \left( n^2 V_{\sigma_1} - \beta n \widetilde{f} \omega - \omega^2  \right) u = 0  \\ & u(x_{-}) = u(x_{+}) = 0 
\end{aligned}
\end{align}
We can identify ${\mathcal{P}}_0$ as the linear eingevalue problem (since the potential does not depend on $\omega$). 
\begin{align} 
\label{shr_lin}
{\mathcal{P}}_0 \quad  : \begin{aligned} 
- \frac{\td^2 u}{ \td x^2 }  +   n^2 V_{\sigma_1} u = \omega^2 u \\  u(x_{-}) = u(x_{+}) = 0 
\end{aligned}
\end{align} and ${\mathcal{P}}_1$ is the nonlinear eigenvalue problem \eqref{shreq} that we want to solve. Hence $\beta$ is a nonlinear parameter that represents a transition from the linear eigenvalue problem ${\mathcal{P}}_0$ to the nonlinear eigenvalue problem ${\mathcal{P}}_1$.

Before continuing with the analysis of these problems, we pause to note similarities in the soliton and $\mbox{Kerr-AdS}_4$ case for the construction of quasimodes, the most fundamental being the phenomenon of stable trapping occuring in both. A key difference arises in the extension of results from $\mathcal{P}_0$ to $\mathcal{P}_1$. In the $\mbox{Kerr-AdS}_4$ case \cite{holzegel:2013kna}, the potential has a nonlinear term which is proportional to $\omega^2$, so the whole eigenvalue equation is quadratic in $\omega^2$.  In our case the nonlinearity is $ \sim \omega n$. The difficulty arises from the presence of the eigenvalue $n$ with $\omega$ and the fact that we have terms proportional to both $\omega^2$ and $\omega$ in the equation.  Such problems were encountered in the analysis of quasimodes and stable trapping in black ring spacetimes~\cite{benomio:2018ivy}, and we will follow the strategy developed there. 

\section{Eigenvalues for the linear problem }
\label{sec:linear}
In this section we use  a suitable version of Weyl's law to establish the existence of eigenfunctions for the linear problem ${\mathcal{P}}_0$ defined by \eqref{shr_lin}.  We start by defining a semi classical parameter $ h^{2} = n^{-2} $
and express the problem in the fform  
\begin{align}
\begin{aligned} 
- h^2 \frac{\td^2 u}{ \td x^2 }  +   V_{\sigma_1} u = \kappa u \\  u(x_{-}) = u(x_{+}) = 0 
\end{aligned}
\end{align} where we have defined $\kappa$ to be the eigenvalue i.e., $\kappa := h^2 \omega^2$.
We identify the region $\Omega := [x_{-},x_{+}]$ for the eigenvalue problem through the following lemma.  

\begin{lemma}
	\label{cty_potential}
	Let $V^{min}_{\sigma_1}$ be the local minimum of the potential and let $x_{min}\in (\alpha,\infty) $ be the point where this minimum is attained i.e., $V_{\sigma_1}({x_{min}}) = V^{min}_{\sigma_1}$. Let $c>0$ be sufficiently small so that there exist $x_{-}$ and $x_{+}$ with $x_{-} < x_{min} < x_{+}$ for which, $V^{min}_{\sigma_1} + c = V_{\sigma_1}(x_{-}) = V_{\sigma_1}(x_{+})$ and there are no local maxima of $V_{\sigma_1}$ in $[x_{-},x_{+}]$. Let $E>V_{\sigma_1}^{min}$ such that $E-V_{\sigma_1}^{min} < c$. Then for any sufficiently small constants $\delta, \delta^{'} > 0$ there exists some constant $c^{'}>0$ such that 
	\begin{align}
	|x_{\pm}- x| < \delta^{'} \implies V_{\sigma_1}(x) - \kappa > c^{'}
	\end{align}    
	for all $\kappa \in [E-\delta,E+\delta]$.
\end{lemma} 

\begin{proof}
	The idea behind the above Lemma is illustrated in Fig.\ref{continuity_potential}.
	\begin{figure}[!htb]
		\centering
		\begin{figure}[H]
			\includegraphics[scale=1]{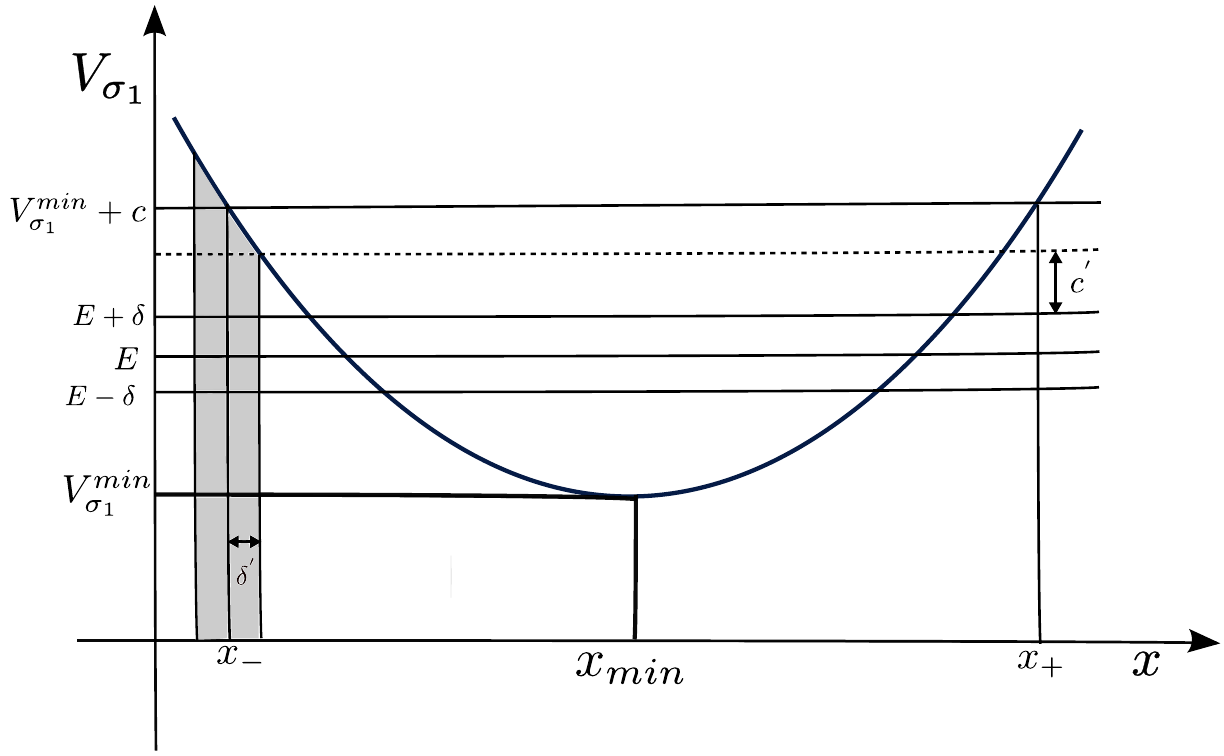}
			\caption{Domain for the eigenvalue problem}
			\label{continuity_potential}
		\end{figure}
	\end{figure}
	We can fix a sufficiently small constant $\delta$ such that $E + \delta < V_{\sigma_1}^{min} + c $. $V_{\sigma_1}(x)$ is continuous at $x_{-}$.  In the following we will establish the result  for $x_{-}$ and the proof for $x_{-}$ replaced by $x_{+}$ follows by a similar argument. For a given $\tilde{\epsilon}$, one can find a $\delta^{'}$ such that, $ |V_{\sigma_1}(x) - V_{\sigma_1}(x_{-}) | < \tilde{\epsilon} $ whenever $|x - x_{-}| < \delta^{'}$. Choose 
	\begin{align}
	\tilde{\epsilon} = \frac{V_{\sigma_1}(x_{-}) - (E+\delta)}{3} \nonumber
	\end{align}
	This means, whenever $|x_{-} - x | < \delta^{'}$, $|V_{\sigma_1}(x_{-}) - V_{\sigma_1}(x)|  < \tilde{\epsilon} $ 
	\begin{align}
	V_{\sigma_1}(x_{-}) - \tilde{\epsilon} <  V_{\sigma_1}(x) < V_{\sigma_1}(x_{-}) + \tilde{\epsilon} \nonumber
	\end{align}
	For $\kappa \in [ E-\delta,E+\delta ]$ we have 
	\begin{align}
	V_{\sigma_1}(x) - \kappa & > V_{\sigma_1}(x_{-}) - \tilde{\epsilon} - \kappa \nonumber \\ 
	& > V_{\sigma_1}(x_{-}) - \tilde{\epsilon} - (E+\delta) \nonumber \\ &= 3 \tilde{\epsilon} - \tilde{\epsilon} = 2 \tilde{\epsilon} \nonumber
	\end{align}
	Setting $c^{'} = 2 \tilde{\epsilon} $ completes the proof.	
\end{proof}

We now state and prove Weyl's law. This allows us to establish the existence of eigenfunctions for the Dirichlet problem in the domain $[x_{-},x_{+}]$. More precisely, this statement proves that the number of eigenvalues $\kappa$ in some small neighborhood scale as $h^{-1}$. The eigenvalue problem with Dirichlet conditions will be denoted by $\mathcal{P}_{D}(x_{-},x_{+})$ and  $N_{\leq E} (\widetilde{\mathcal{P}})$ denotes the number of eigenvalues of the problem $\widetilde{\mathcal{P}}$ which are less than or equal to $E$.

\begin{theorem}[\textbf{Weyl's law}]
	\label{weyls_law}
	Consider the eigenvalue problem $\mathcal{P}_{D}(x_{-},x_{+})$. Let $E$ be an energy level such that $E - V_{\sigma_1}^{min}$ is sufficiently small and $E - V_{\sigma_1}^{min} > \delta$ for some fixed positive constant $\delta$ such that $E + \delta < V_{\sigma_1}^{min} + c$ with the constant $c>0$ introduced in Lemma \ref{cty_potential}. Then the number of eigenvalues of the problem $\mathcal{P}_{D}(x_{-},x_{+})$ less than $E$, denoted by $N_{\leq E} (\mathcal{P}_{D}(x_{-},x_{+}) )$, satisfies the following estimate called Weyl's law.		
	\begin{align}
	N_{\leq E} (\mathcal{P}_{D}(x_{-},x_{+}) ) \sim \mathcal{Q}_{E,h}  
	\end{align} 
	where
	$$ \mathcal{Q}_{E,h} :=  \di \frac{1}{h \pi}\int_{x_{-}}^{x_{+}} \di \sqrt{ E - \di V_{\sigma_1}(x^*)} \chi_{\{ V_{\sigma_1} \leq E \}} \mathrm{d} x^{*} .$$
\end{theorem}
We will also establish the following result which estimates the number of eigenvalues for the problem for $\mathcal{P}_{D}(x_{-},x_{+})$ lying in a $\delta$ interval of $E$.  
\begin{theorem} \label{weyl-appl}
	Let $N[E-\delta,E+\delta]$ denote the number of eigenvalues of $\mathcal{P}_{D}(x_{-},x_{+})$ lying in the interval $[E-\delta,E+\delta]$. Then $N[E-\delta,E+\delta]$ satisfies Weyl's law i.e., 	\begin{equation}
	N[E-\delta,E+\delta] \sim \mathcal{Q}_{E+\delta,h} - \mathcal{Q}_{E-\delta,h}
	\end{equation}			
\end{theorem}
We will prove this with the following two Lemmas which give upper and lower bounds for $N_{\leq E} (\mathcal{P}_{D}(x_{-},x_{+}) )$. These bounds will be explicitly calculated. We first partition $[x_{-},x_{+}]$ into $k$ intervals $[x^{i}_{-}, x^{i}_{+}]$ where
\begin{align}
x^{i}_{-} = {x_{-} + (i-1) \gamma } & \, \mbox{  and  } \, x^{i}_{+} = {x_{-} + i \gamma } \quad \mbox{ where } \gamma  = \frac{x_{+}-x_{-}}{k}\,, 
\end{align}
and define $k$ Dirichlet problems in each $[x^i_{-},x^i_{+}]$. The Dirichlet problems  $\mathcal{P}^i_{D}$ for $i=1,2...k$ are stated below 
\begin{align}
\begin{aligned}
- h^2 \frac{\td^2 u}{ dx^2 } +  V_{\sigma_1}    u & = \kappa u  \\
u({x^i_{-} })  = u({x^i_{+}}) &= 0
\end{aligned}
\end{align}
We next define $k$ Neumann problems $\mathcal{P}^i_{N}$ analogously. $\mathcal{P}^i_{D}$ and $\mathcal{P}^i_{N}$ will serve as two comparison problems for estimating $N_{\leq E} ({\mathcal{P}_D}(x_{-},x_{+}))$ through lower and upper bounds respectively. We start with the following Lemma which gives a lower bound through the $k$ Dirichlet problems $\mathcal{P}^i_{D}$.

\begin{lemma}[\textbf{Lower bound}]
	\label{lowerb} 
	The number of eigenvalues of the problem ${\mathcal{P}_D}(x_{-},x_{+})$ less than $E$ i.e., $N_{\leq E} ({\mathcal{P}_D}(x_{-},x_{+}))$ satisfies 
	\begin{align}
	\sum_{i}^{k} N_{\leq E} (\mathcal{P}^i_{D}) \leq N_{\leq E} ({\mathcal{P}_D}(x_{-},x_{+})) \, .
	\end{align}
\end{lemma}

\begin{proof}
	The proof relies on the variational characterization of eigenvalues using the min-max principle. The smallest eigenvalue of the problem ${\mathcal{P}_D}(x_{-},x_{+})$ can be characterized by
	\begin{align}
	\kappa_1 = \inf_{ \substack{
			u \in H^1_{0}([x_{-},x_{+} ]) \\ 		\norm{u}_{L^2} \neq 0  \quad } }
	\frac{ \int_{x_{-}}^{x_{+} }  (h^2 |\partial_x u|^2 + V_{\sigma_1}(x) |u|^2) \td x }{ \norm{u}^2_{L^2} }
	\end{align}
	The $n$-th eigenvalue of ${\mathcal{P}_D}(x_{-},x_{+})$ (this is not to be confused with the integer $n$ appearing in the separation of variables \eqref{antsatz}) can be characterized by
	\begin{align}
	\kappa_n = \inf_{ \substack{
			\{ u_1, u_2, \ldots, u_n \},\, u_m \in H^1_{0}([x_{-},x_{+} ]) \\
			\norm{u_m}_{L^2} \neq 0 ,\,
			\langle u_m  ,  u_j \rangle = 0 \,\, \forall m \neq j }} \,\,
	\max_{m \leq n}
	\frac{ \int_{x_{-}}^{x_{+} } (h^2 |\partial_x u_m|^2 + V_{\sigma_1}(x) |u_m|^2) \td x }{ \norm{u_m}^2_{L^2} }
	\end{align}
	Similarly, we can characterize the eigenvalues for  $\mathcal{P}^i_{D}$, denoted by $\lambda^i_n$ as
	\begin{align}
	\lambda^i_n = \inf_{ \substack{
			\{u_1, u_2, \ldots, u_n\},\, u_m \in {H}^1_{0}([x^i_{-},x^i_{+} ]) \\
			\norm{u_m}_{L^2} \neq 0 ,\,
			\langle u_m  ,  u_j \rangle = 0 \,\, \forall m \neq j }}\,\,
	\max_{m \leq n}
	\frac{ \int_{x^i_{-}}^{x^i_{+} } (h^2 |\partial_x u_m|^2 + V_{\sigma_1}(x) |u_m|^2) \td x }{ \norm{u_m}^2_{L^2} }
	\end{align}
We can see from the variational characterization that $\kappa_n \leq \lambda^i_n  $. By arranging all the eigenvalues $\lambda^i_n$ into a single non-decreasing sequence $\lambda_n$, we can deduce the following : 
	\begin{equation}
	\kappa_n \leq \lambda_n.
	\end{equation} 
	To see this, let $f_n$ be the eigenfunctions corresponding to $\lambda_n$. $f_n$ can be extended to $[x_{-},x_{+}]$ by setting them to vanish outside the corresponding $[x^i_{-},x^i_{+}]$. These $n$ functions are orthogonal in $H^1_0[x_{-},x_{+}]$ either because they are eigenfunctions supported in different regions or because they are different eigenfunctions (with the same or different eigenvalues) to the same problem, which makes them orthogonal \cite{Keir:2014oka}. Hence we have $\kappa_n \leq \lambda_n$ which proves the inequality. 
\end{proof}  
From the $k$ Neumann problems $\mathcal{P}^i_{N}$ and their corresponding eigenvalues $\mu^i_n$, we have the following Lemma.
\begin{lemma}[\textbf{Upper bound}]
	\label{upperb}
	The number of eigenvalues of the problem ${\mathcal{P}_D}(x_{-},x_{+})$ less than $E$ i.e., $N_{\leq E} ({\mathcal{P}_D}(x_{-},x_{+}))$ satisfies 
	\begin{align}
	N_{\leq E} ({\mathcal{P}_D}(x_{-},x_{+})) \leq \sum_{i}^{k} N_{\leq E} (\mathcal{P}^i_{N}).  
	\end{align}
\end{lemma}
\begin{proof}
	The eigenvalues $\mu^i_n$ can be characterized as
	\begin{align}
	\mu^i_n =  \inf_{ \substack{
			\{u_1, u_2, \ldots, u_n\} ,\, 
			u_m \in \widetilde{H}^{1} \left([x_{-},x_{+}]\right) \\ \norm{u_m}_{L^2} \neq 0 ,\, \langle u_m  ,  u_j \rangle = 0 \,\, \forall m \neq j }} \,\,
	\max_{m \leq n} 
	\frac{ \di \sum_{i=1}^{k} \int_{x^i_{-}}^{x^i_{+} } (h^2 |\partial_x u_m|^2 + V_{\sigma_1}(x) |u_m|^2) \td x }{ \norm{u_m}^2_{L^2} }
	\end{align}
	where, $$ \widetilde{H}^{1} \left([x_{-},x_{+}]\right) = \left\{ u_m \in L^2([x_{-},x_{+}])  | u_m \in H^1 ([x^i_{-},x^i_{+}]) \mbox{ for all } i \right\} $$
	Similar to the previous case, we arrange them in a single non-decreasing sequence $\mu_n$. We observe that $ H^{1}_0([x_{-},x_{+}]) \subset \widetilde{H}^{1}([x_{-},x_{+}])  $ which implies that $ \mu^i_n \leq \kappa_n $. In particular this means, $\mu_n \leq \kappa_n$ which completes the proof. Since we are in one dimension, the $H^1$ spaces mentioned here in fact embed into Holder spaces $C^{0,1/2}$. 
\end{proof}

\begin{proof}[Proof of Theorem \ref{weyls_law} (Weyl's law)] 
	To compute explicit bounds for  $N_{\leq E} ({\mathcal{P}_D}(r_{-},r_{+})) $ we consider the following sets of problems. 
	\begin{itemize}
		\item $\widetilde{\mathcal{P}}^i_{D}$ : Problems  ${\mathcal{P}}^i_{D}$ where the potential $V_{\sigma_1}$ is replaced by its maximum value (say $V^{i}_{+}$) in the interval $[x^i_{-},x^i_{+}]$.
		\item  $\widetilde{\mathcal{P}}^i_{N}$ : Problems  ${\mathcal{P}}^i_{N}$ where the potential $V_{\sigma_1}$ is replaced by its minimum value (say $V^{i}_{-}$) in the interval $[x^i_{-},x^i_{+}]$.
	\end{itemize}
	The bounds for $N_{\leq E} ({\mathcal{P}_D}(x_{-},x_{+}))$ in Lemmas \ref{lowerb} and \ref{upperb} hold when ${\mathcal{P}}^i_{D}$ and  ${\mathcal{P}}^i_{N}$ are replaced by $\widetilde{\mathcal{P}}^i_{D}$ and $\widetilde{\mathcal{P}}^i_{N}$ respectively. These problems can be solved exactly as the potential is just a constant in the interval. The number of eigenvalues of $\widetilde{\mathcal{P}}^i_{D}$ with energy less than or equal to $E$ is given by, 
	
	\begin{align}
	\begin{aligned}
	N_{\leq E} (\widetilde{\mathcal{P}}^i_{D}) & =  \floor*{\frac{\gamma \sqrt{E - V^{i}_{+}}}{h \pi} \chi_{\{V^{i}_{+} \leq E\}} } \\ 
	\sum_{i=1}^{k} N_{\leq E} (\widetilde{\mathcal{P}}^i_{D}) & = \sum_{i=1}^{k} \floor*{\frac{\gamma \sqrt{E - V^{i}_{+}}}{h \pi} \chi_{\{V^{i}_{+} \leq E\}}}  \\ 
	& = \sum_{i=1}^{k} \left(\frac{\gamma \sqrt{E - V^{i}_{+}}}{h \pi} \chi_{\{V^{i}_{+} \leq E\}}\right) + \mathcal{O}(k) 
	\end{aligned}
	\end{align} 
	Similarly for $\widetilde{\mathcal{P}}^i_{N}$, we have, 
	\begin{align}
\begin{aligned}
	N_{\leq E} (\widetilde{\mathcal{P}}^i_{N}) & =  \floor*{\frac{\gamma \sqrt{E - V^{i}_{-}}}{h \pi} \chi_{\{V^{i}_{-} \leq E\}} } + 1 \\ 
	\sum_{i=1}^{k} N_{\leq E} (\widetilde{\mathcal{P}}^i_{N}) & = \sum_{i=1}^{k} \floor*{\frac{\gamma \sqrt{E - V^{i}_{-}}}{h \pi} \chi_{\{V^{i}_{-} \leq E\}}} + k \\ & 
	= \sum_{i=1}^{k} \left(\frac{\gamma \sqrt{E - V^{i}_{-}}}{h \pi} \chi_{\{V^{i}_{-} \leq E\}}\right) + \mathcal{O}(k)
\end{aligned}
	\end{align} 
	Based on Lemmas \ref{lowerb} and \ref{upperb}, $N_{\leq E}
	({\mathcal{P}_D}(x_{-},x_{+}))$ satisfies	
	\begin{align} 
	\sum_{i=1}^{k} N_{\leq E} (\widetilde{\mathcal{P}}^i_{D}) \leq N_{\leq E} ({\mathcal{P}_D}(r_{-},r_{+})) \leq \sum_{i=1}^{k} N_{\leq E} (\widetilde{\mathcal{P}}^i_{N}) 
	\end{align}
	which becomes,
	\begin{align}
	\sum_{i=1}^{k} \left(\frac{\gamma \sqrt{E - V^{i}_{+}}}{h \pi} \chi_{\{V^{i}_{+} \leq E\}}\right) + \mathcal{O}(k) \leq 
	N_{\leq E} ({\mathcal{P}_D}(x_{-},x_{+})) \leq \sum_{i=1}^{k} \left(\frac{\gamma \sqrt{E - V^{i}_{-}}}{h \pi} \chi_{\{V^{i}_{-} \leq E\}}\right) + \mathcal{O}(k) .
	\nonumber
	\end{align} 
	If we let the number of partitions go to infinity as $h \to 0$ such that $ k(h) = o (1/h)$, the sums converge as a Riemann sum and the error terms are of order $o(1/h)$. We can then express $ N_{\leq E} ({\mathcal{P}_D}(x_{-},x_{+}))$ as
	\begin{align}
	N_{\leq E} ({\mathcal{P}_D}(x_{-},x_{+})) \sim 
	\frac{1}{h \pi}\int_{x_{-}}^{x_{+}} \sqrt{E - V_{\sigma_1}(x^*)} \chi_{\{V_{\sigma_1} \leq E\}} \td x^*  .
	\label{weyl-proof}
	\end{align}
	This proves Theorem \ref{weyls_law}.
\end{proof}

\begin{proof}[Proof of Theorem \ref{weyl-appl}]
	This follows by computing $ N_{\leq E + \delta} ({\mathcal{P}_D}(x_{-},x_{+}))$  and $N_{\leq E - \delta} ({\mathcal{P}_D}(x_{-},x_{+})) $ from \eqref{weyl-proof}. 
\end{proof}

\section{Eigenvalues for the nonlinear problem}
\label{sec:nonlinear}

We now turn to establishing the existence of eigenvalues for the radial equation, which as discussed earlier is nonlinear in the `energy' $\omega$, namely the ODE $\mathcal{P}_1 $ which reads,
\begin{equation}
	-\frac{\td^2 u}{ \td x^2 } +  V u = 0 , \qquad  u(x_{-}) = u(x_{+}) = 0. \label{nonlinear}
\end{equation}
Here $V$ is the nonlinear potential defined by $ V := n^2 V_{\sigma_1} - n \widetilde{f} \omega - \omega^2  $. The strategy is to prove the existence of eigenvalues of \eqref{nonlinear} through continuity arguments. The potential is a complicated rational function of the rescaled radial variable $x$ explicitly given by 
\begin{align}
	V = \frac{ \left( {x
		}^{2} -{\alpha}^{2} \right)^{-1} }{{16\, \left( {\alpha
			}^{6}+{\alpha}^{2}{x}^{2}+{x}^{4} \right) ^{2}  }} \left\{ \right. & \left. \left( 129 -8\,n\omega \right) {\alpha}^{14} + \left( 128 -
	\left(129- 8\,n\omega \right) {x}^{2} \right) {\alpha}^{12} - \left(	8\,n\omega\,{x}^{2}+128 \right) {\alpha}^{10} \right. \nonumber \\ & \left.  +128\ {x}^{2} {\alpha}^{8}
	+ \left(  \left( 8\,n\omega-144 \right) {x}^{6}+128\,{x}^{4}-128\,{x}^
	{2} \right) {\alpha}^{6}   \right.  \\ \nonumber & \left.    - \left( 144\,{x}^{6}+128\,{x}^{4} \right) {
		\alpha}^{4}+144\,{x}^{6}{\alpha}^{2}+144\,{x}^{8}\right\} - \omega^2 . 
\end{align}
For the nonlinear problem we want to reproduce the setting of the linear problem $\mathcal{P}_0$. We begin by  verifying that there is  still a trapped region. From the defintion of $V$ above, this would amount to checking that there is a region where $V$ has a negative minimum. Here we are interested in determining the existence of eigenvalues close to 0 (as opposed to eigenvalues close to $E$ in the linear case). Lemma \ref{cty_potential} identified  such a region for $\mathcal{P}_0$. Here we state a nonlinear version i.e., Lemma \ref{npcty} (following \cite{benomio:2018ivy}) which identifies the corresponding $\Omega$ for $\mathcal{P}_1$. In the following proposition, we list some properties of $V$ which will be useful in proving Lemma \ref{npcty}. Elements of the proofs in Proposition \ref{propnonlin} and Lemma \ref{npcty} which involve the structure of $V$ will be illustrated with some plots owing to the complicated expression of $V$. To emphasize the dependence of $V$ on $\omega$ and $n$, we denote the nonlinear potential as $V_{(\omega,n)}$ rather than $V$.

\begin{prop}
	[Properties of ${V}_{(\omega,n)}$] 	\label{propnonlin}
	Consider $\omega \in \mathbb{R}$ and $n \in \mathbb{Z}$. 
	\begin{enumerate}
		\item If $\omega=0$, $ {V}_{(0,n)}$ is always positive and does not admit any real roots. 
		\item There exists a pair $(\omega_0,n_0)$ such that ${V}_{(\omega_0,n_0)}$ admits three distinct real roots $ x^{\omega_0}_1, x^{\omega_0}_2 $ and $ x^{\omega_0}_3 $ such that $ {V}_{(\omega_0,n_0)}$ has a local minimum at $ x^{\omega_0}_{min}$ with $x^{\omega_0}_1 < x^{\omega_0}_{min} < x^{\omega_0}_2 < x^{\omega_0}_3 $.
		\item Consider a pair $(\omega_0,n_0)$ for which ${V}_{(\omega_0,n_0)}$ admits three distinct real roots. There exist $\mathcal{E^{-}}$ and $\mathcal{E}^{+}$ such that 
		\begin{enumerate}
			\item 	 $\omega_0 \in (\mathcal{E^{-}},\mathcal{E}^{+}) $ and for any $ \omega \in (\mathcal{E^{-}},\mathcal{E}^{+})$ ${V}_{(\omega,n_0)} $ admits three distinct real roots.
			\item   Let $\omega_1$ and $\omega_2$ be two such values with $\{ x^{\omega_1}_1,x^{\omega_1}_2, r^{\omega_1}_3 \}$ and $\{x^{\omega_2}_1,x^{\omega_2}_2,x^{\omega_2}_3 \}$ being the corresponding roots. If $\omega_1 > \omega_2$, then $(x^{\omega_2}_1,x^{\omega_2}_2) \subsetneq (x^{\omega_1}_1,x^{\omega_1}_2) $. 	  
		\end{enumerate}
	\end{enumerate}
\end{prop}
\begin{proof}
	We know that $ \lim\limits_{x \to \alpha} {V} = \infty $ and $\lim\limits_{x \to \infty} {V} = -\omega^2 $. Hence the potential admits at least one real root.
	We note that the potential is invariant under $(\omega,n) \to (-\omega,-n) $. Hence we assume that ${\omega} > 0 $ and only discuss the cases $n \in \mathbb{Z}^+ $ and $n \in \mathbb{Z}^-$ where needed.\begin{enumerate}
		\item With $\omega =0$, we have ${V}_{(0,n)} = n^2 V_{\sigma_1} > 0 $ as can be seen from Fig.\ref{sol_potential}.
		\item \begin{enumerate}
			\item Case 1 : $n \in \mathbb{Z}^+$. For $ {\omega} \in [1.465n,1.485n]$, $ {V}_{({\omega},n)} $ admits two roots. This can be seen in the plots below (Figs.\ref{NPnp1} and \ref{NPnp2}).
			\begin{figure}[H]
				\label{rootsposn}
				\subfloat[ $\di \frac{{V}}{n^2}$ for ${\omega}=1.465n$ \label{NPnp1} ]{\includegraphics[scale=2.7]{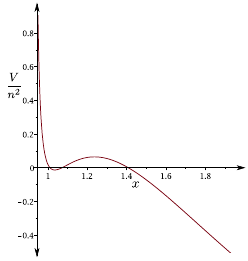}}
				\subfloat[$\di \frac{{V}}{n^2}$ for ${\omega}=1.485n$ \label{NPnp2}]{\includegraphics[scale=2.7]{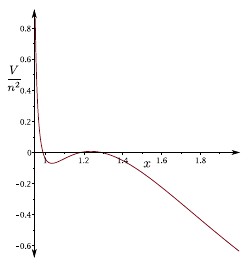}}
				\caption{Nonlinear potential for $n\in\mathbb{Z}^{+}$}
			\end{figure}
			
			\item Case 2 : $n \in \mathbb{Z}^-$. For $ {\omega} \in [1.35|n|,1.415|n|]$, we can see from Figs.\ref{NP1} and \ref{NP2} that ${V}_{({\omega},n)} $ admits two roots. 			
			\begin{figure}[htb!]
				\subfloat[ $\di \frac{{V}}{n^2}$ for $\omega=1.35|n|$ \label{NP1} ]{\includegraphics[scale=2.7]{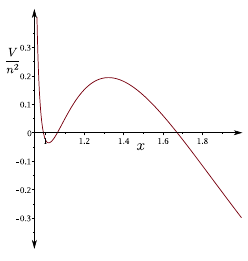}}
				\subfloat[$\di \frac{{V}}{n^2}$ for $\omega=1.415|n|$ \label{NP2}]{\includegraphics[scale=2.7]{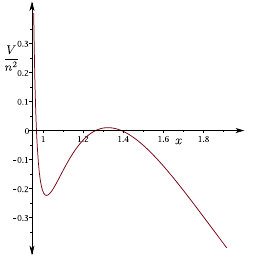}}
				\caption{Nonlinear potential for  $n\in\mathbb{Z}^{-}$}
			\end{figure}			
		\end{enumerate}
		\item As a consequence of (2) above, for $n \in \mathbb{Z}^{+}$, the choice $ \mathcal{E}^{-} = 1.47 $ and $\mathcal{E}^{+} = 1.48 $ satisfies the condition $(a)$. For $(b)$, in the following plot (Fig.\ref{regposn}), as $ \omega $ increases in $(\mathcal{E^{-}},\mathcal{E^{+}}) $, the corresponding interval $(x_{-}^{\omega},x_{+}^{\omega} )$ also increases. For the case $ n \in \mathbb{Z}^{-}$, we observe that $ \widetilde{f} <0$. Hence, from the following expression for the nonlinear potential,
		\begin{align}
			{V}   = n^2 V_{\sigma_1}  - n \widetilde{f} {{\omega}} - {{\omega}^2} 
		\end{align}
		the increase (decrease) of the interval $(x_{-}^{\omega},x_{+}^{\omega} )$ with increase (decrease) in $\omega$ follows. This can also be seen in Fig.\ref{regnegn}
	\end{enumerate}	
	\begin{figure}[H]
		\subfloat[ $n \in \mathbb{Z}^{+} $ \label{regposn} ]{\includegraphics[scale=1]{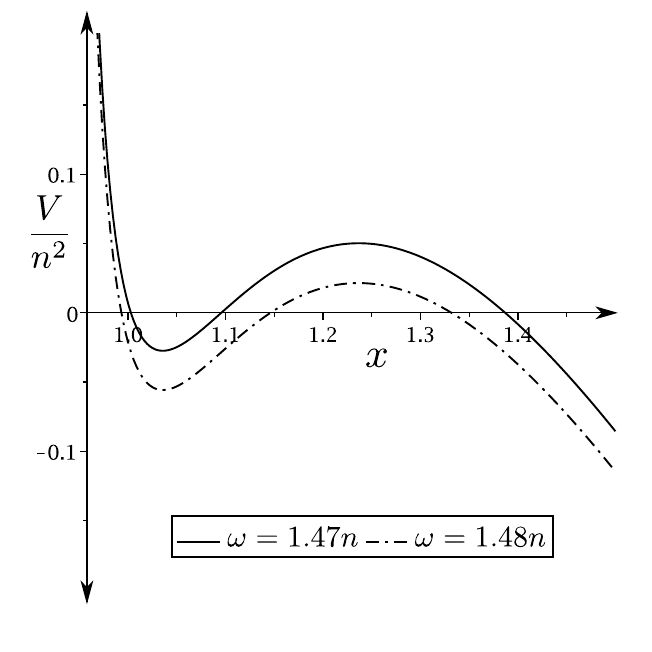}}
		\subfloat[$n \in \mathbb{Z}^{-} $ \label{regnegn}]{\includegraphics[scale=1]{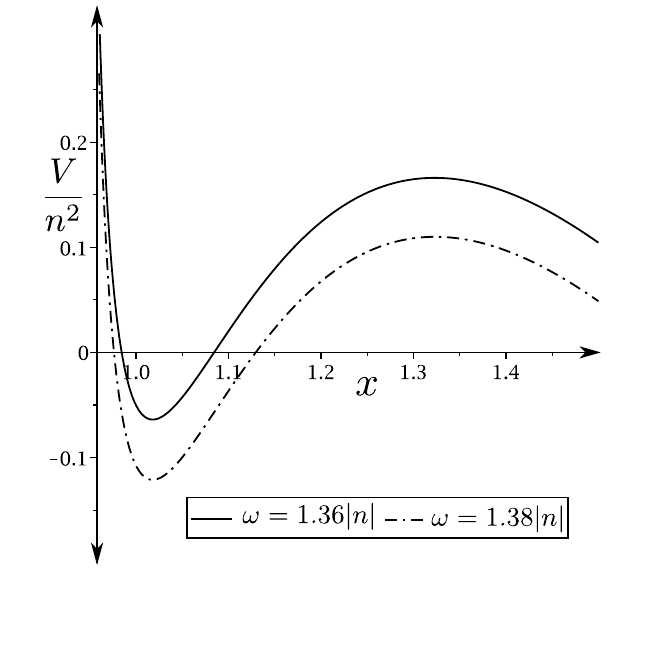}} \caption{ Properties of $(x^\omega_{-},x^\omega_{+})$ }
		\label{regprop}
	\end{figure}
\end{proof}
\begin{lemma}
	\label{npcty}
	Let ${V}_{min}$ be the minimum of the nonlinear potential ${V}_{(\omega,n)}$. Let $x_{min} \in (\alpha,\infty)$ such that ${V}_{(\omega,n)} (x_{min}) = {V}_{min} $. Consider some constant $\mathcal{E} >0$ such that ${V}_{(\mathcal{E}n,n)}$ has a local minimum and there exists $x^{\mathcal{E}}_{-}$ and  $x^{\mathcal{E}}_{+}$ satisfying $ x^{\mathcal{E}}_{-} < x_{min} < x^{\mathcal{E}}_{+}$ for which $ {V}_{(\mathcal{E} n,n)} (x^{\mathcal{E}}_{-}) = {V}_{(\mathcal{E} n,n)} (x^{\mathcal{E}}_{+}) = 0  $ and there are no local maxima of $ {V}_{(\mathcal{E}n,n)}$ in $(x^{\mathcal{E}}_{-},x^{\mathcal{E}}_{+})$. Let $E>0$ be an energy level such that $E < \mathcal{E}$ and ${V}_{(En,n)}$ has a local minimum and there exists constants $x^E_{-}$ and  $x^E_{+}$ with the same properties as $x^{\mathcal{E}}_{-}$ and  $x^{\mathcal{E}}_{+}$ respectively but now with respect to ${V}_{(En,n)} $. Then for sufficiently small constants $\delta, \delta^{'} >0$ there exists a constant $c>0$ such that 
	\begin{align}
		|x - x^{\mathcal{E}}_{\pm}| < \delta^{'} \implies \frac{1}{n^2} {V}_{(\kappa n,n)} > c 
	\end{align}
	for all $ \kappa \in \mathbb{R}$ satisfying $|\kappa^2 - E^2| \leq \delta$. In addition, for the linear problem Lemma \ref{cty_potential} holds for $ \left. \di  \frac{n^2 V_{\sigma_1 } - \omega^2 }{n^2} \right|_{\omega = En} $. 
\end{lemma}

\begin{figure}[htb!]
	\includegraphics[scale=1.25]{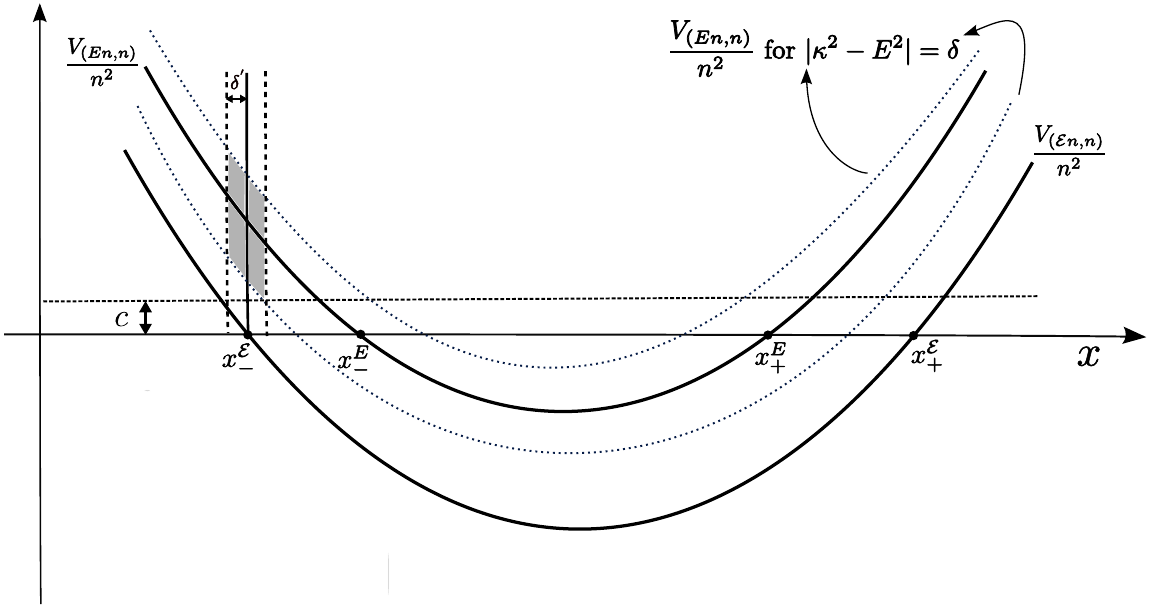}
	\caption{Domain for the nonlinear eigenvalue problem}
	\label{cty_nonlinear_potential}
\end{figure}

\begin{proof}
	We observe the following as consequences of Proposition \ref{propnonlin}.
	\begin{enumerate}
		\item There exists $E $ and $\mathcal{E}$ for which the potential admits three distinct real roots with a local minimum as shown in Fig \ref{cty_nonlinear_potential}.  
		\item $E < \mathcal{E} \implies (x^E_{-}, x^{E}_{+}) \subsetneq (x^\mathcal{E}_{-},x^{\mathcal{E}}_{+})  $.
	\end{enumerate} 
	Hence, $\mathcal{E}$ and $E$ with the desired properties exist.  We can see that $V_{(En,n)}$ has no local maxima in $(x^\mathcal{E}_{-},x^{\mathcal{E}}_{+})$. $V_{(En,n)}(x^E_{-}) = 0$, so for $x \in [ x^\mathcal{E}_{-}, x^E_{-})$, $ \di \frac{V_{(En,n)}(x)}{n^2} >0$. Then it follows that there exists $\delta^{'} >0$ such that for $ | x - x^\mathcal{E}_{-} | < \delta^{'} $, $ \di \frac{V_{(En,n)}(x)}{n^2} > 0$. Since $ \di \frac{V_{(En,n)}(x)}{n^2}$ is also continuous as a function of $E$, there exists some $\delta>0$ such that for $|\kappa^2 - E^2| \leq \delta $,  $ \di \frac{V_{(\kappa n,n)}(x)}{n^2} > c$ for some constant $c>0$. For the final part of the Lemma, let us refer to Fig. \ref{lin_nonlin1} and 	Fig. \ref{lin_nonlin2}.

		\begin{figure}[h]
			\subfloat[ $E = 1.405$ \label{lin_nonlin1} ]{\includegraphics[scale=1]{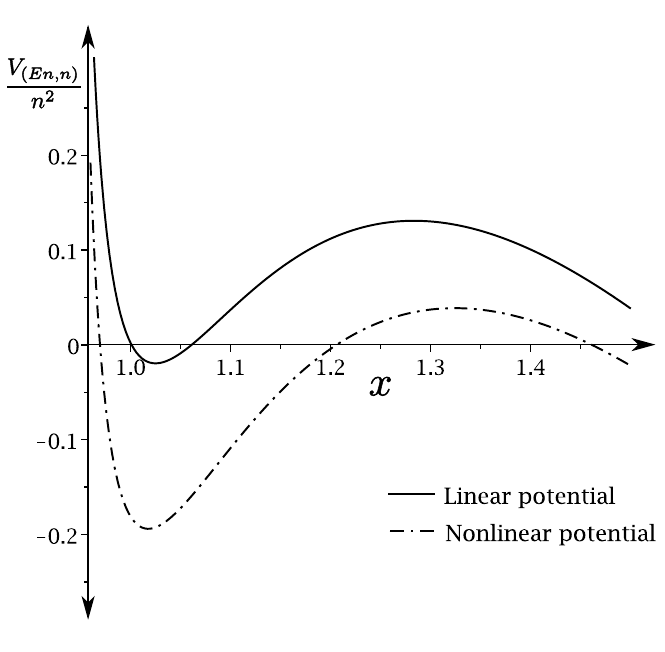}}
			\subfloat[$E = 1.415$ \label{lin_nonlin2}]{\includegraphics[scale=1]{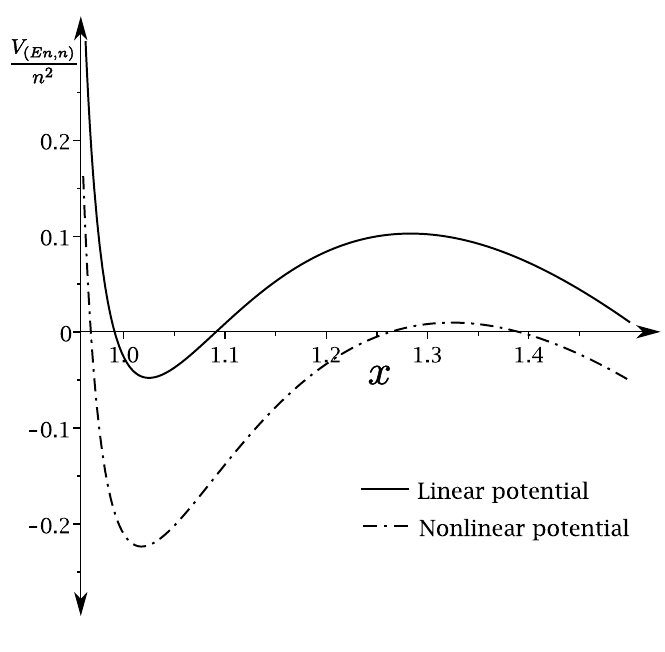}}
\caption{Continuity for linear and nonlinear potentials}
		\end{figure}
	
	For the relation to the linear potential, we observe that for $n \in \mathbb{Z}^{-}$ there exists $E \in [1.405,1.415]$ such that Lemma \ref{cty_potential} holds.  	
\end{proof}
Hence from this stage onwards,  we will only consider  $n \in \mathbb{Z}^{-}$ as it is sufficient for our construction of quasimodes.  We use $-|n|$ instead of $n$ in expressions to indicate this sign choice. 
\subsection{Lower bound for $\omega^2$} In this section, we establish a lower bound for the eigenvalues of $\mathcal{P}_1$. Consider the following family of eigenvalue problems 
\begin{align}
	\label{np_family}
	\begin{aligned}
		& \mathcal{Q}(\beta,\omega)  u  = \Lambda(\beta,\omega)  \\ \qquad 
		& u(x_{-})  = u(x_{+}) = 0 \\ 
		\mbox{ where }  \mathcal{Q}(\beta,\omega) & := 
		- \frac{\td^2 u}{ \td x^2 } +  n^2 V_{\sigma_1} + \beta |n| \widetilde{f} \omega - \omega^2
	\end{aligned} 
\end{align}
We prove that if the $j^{\mbox{th}}$ eigenvalue of $\mathcal{Q}(\beta,\omega)$ i.e., $ \Lambda_j (\beta,\omega) $ is zero, then the corresponding $\omega_{\beta,n}$ satisfies certain properties. Let $u_j(\beta,\omega)$ be a normalized eigenfunction in $[x_{-},x_{+}]$ associated to the eigenvalue $\Lambda_j(\beta,\omega)$. Then 
\begin{align}
	\Lambda_j(\beta,\omega) = \int_{x_{-}}^{x_{+}} u_j(\beta,\omega) \, \mathcal{Q}(\beta,\omega) \, u_j(\beta,\omega) \, \td x = 0.
\end{align}
We have the following lemma which gives a lower bound on $\omega_{\beta,n}$.
\begin{lemma}
	\label{bound_omega}
	Let $u_j(\beta,n)$ be a nontrivial eigenfunction of \eqref{np_family}, then the following hold for sufficiently large $n$ and $\beta \in [0,1]$ :
	\begin{enumerate}
		\item  $ \omega_{\beta,n} \neq 0 $
		\item $ \omega_{\beta,n} \neq o(n) $.
	\end{enumerate} 
\end{lemma}
\begin{proof}
	\begin{enumerate}
		\item If $\omega_{\beta,n} =0$ we have
		\begin{align}
			\begin{aligned}	\int_{x_{-}}^{x_{+}} & u_j(\beta,0)  \mathcal{Q}(\beta,\omega) u_j(\beta,\omega) \, \td x \\ &= \int_{x_{-}}^{x_{+}}  \left.  - u_j(\beta,0) \frac{\td^2 u_j(\beta,0)}{ \td x^2 } + (n^2 V_{\sigma_1} u^2_j(\beta,0)  +   \beta |n| f \omega u^2_j(\beta,0) - \omega^2 u^2_j(\beta,0) ) \right|_{\omega =0} \, \td x \\ &=  \int_{x_{-}}^{x_{+}}   - u_j(\beta,0) \frac{\td^2 u_j(\beta,0)}{ \td x^2 } + (n^2 V_{\sigma_1} u^2_j(\beta,0)  ) \, \td x\\ & = \int_{x_{-}}^{x_{+}}   \left| \frac{\td u_j(\beta,0)}{ \td x } \right|^2 \, \td x + \int_{x_{-}}^{x_{+}} n^2 V_{\sigma_1} u_j^2(\beta,0) \td x .  \end{aligned} 
		\end{align}	
		The first integral is positive and $V_{\sigma_1} >0$. Hence $ \Lambda_j(\beta,0) \neq 0$ which concludes the proof of the first part.
		\item Suppose that $\omega_{\beta,n} = o(n)$.  We proceed with the same steps as above and arrive at, 
		\begin{align}
			\int_{x_{-}}^{x_{+}} u_j(\beta,0)  \mathcal{Q}(\beta,\omega) \, u_j(\beta,\omega) \, \td x  & = \int_{x_{-}}^{x_{+}}   \left| \frac{\td u_j(\beta,0)}{ \td x } \right|^2 \td x + \int_{x_{-}}^{x_{+}} (n^2 V_{\sigma_1}   +  \beta |n| \widetilde{f} \omega  - \omega^2) u^2_j(\beta,\omega) \, \td x \nonumber 
		\end{align}	
		With $\omega_{\beta,n} = o(n)$, we have
		\begin{align}
			n^2 V_{\sigma_1}  +   \beta |n| \widetilde{f} \omega  - \omega^2 \geq n^2 V_{\sigma_1} - C \left[  1- \beta \widetilde{f}\right] n^2
		\end{align}	
		The right hand side is positive when $ C$ is sufficiently small making the second term in the integral positive implying that $\omega  \neq o(n)$. In particular we note here that the above result holds for $\beta =0$ which is the case for the linear eigenvalue problem. 	 
	\end{enumerate}	
\end{proof}

\begin{cor}
	As a consequence of Lemma \ref{bound_omega}, we conclude that given the existence of eigenvalues $\omega_{\beta,n}$ for sufficiently large $n^2$,
	\begin{align} \nonumber
		\omega_{\beta,n}^2 \geq \mathcal{O}(n^2)
	\end{align}i.e., there exists a positive constant $C_\beta$ independent of $n^2$ such that \begin{align} \nonumber
		\omega_{\beta,n}^2 \geq C_{\beta} \, n^2. 
	\end{align}
\end{cor}

\subsection{Eigenvalues for $\beta \neq 0$}
\begin{lemma}
	\label{implicit}
	Let $\beta_0 \in [0,1]$, $\omega_{\beta_0,n} >0$ and $n \in \mathbb{Z}^{-}$ be such that the $j^{\mathrm{th}}$ eigenvalue of $ \mathcal{Q}(\beta_0, \omega_{\beta_0,n})$ is zero. Then for sufficiently large $n^2$, there exists a constant $\epsilon >0$ (independent of $\beta_0$) such that there is a differentiable function $\omega_{\beta,n}(\beta)$ such that the $n^{\mathrm{th}}$ eigenvalue of $\mathcal{Q}(b,\omega_{\beta,n})$ is zero for any $\beta \in (\max (0,\beta_0-\epsilon),\beta_0+\epsilon)$.  
\end{lemma}
\begin{proof}
We start with the expression for the $j^{\mbox{th}}$ eigenvalue $ \Lambda_j(\beta,\omega) $ : 
\begin{align}
	\Lambda_j(\beta,\omega) = \int_{x_{-}}^{x_{+}} u_j(\beta,\omega) \, \mathcal{Q}(\beta,\omega) \, u_j(\beta,\omega) \, \td x .
\end{align}
We assume that the $j^{\mathrm{th}}$ eigenvalue is zero. $ \Lambda_j(\beta_0,\omega_{\beta_0,n}) = 0$ gives an implicit relation between $\beta$ and $\omega_{\beta,n}$. In a neighbourhood of $\beta_0$, the implicit function theorem provides necessary and sufficient conditions for the existence of $\omega_{\beta,n}(\beta)$. We have, 

\begin{align}
	\frac{\partial \Lambda_j}{\partial \omega} (\beta_0, \omega_{\beta_0,n}) &= \int_{x_{-}}^{x_{+}} u_j(\beta,\omega) \, \frac{\partial}{\partial \omega} \left( n^2 V_{\sigma_1}   +   \beta |n| \widetilde{f} \omega  - \omega^2 \right)  \, u_j(\beta,\omega) \, \td x  \nonumber \\ 
	& = \int_{x_{-}}^{x_{+}} u_j(\beta,\omega) \,  \left(      \beta |n| \widetilde{f}   - 2\omega \right)  \, u_j(\beta,\omega) \, \td x .   \nonumber \\ 
	\frac{\partial \Lambda_j}{\partial \beta} (\beta_0, \omega_{\beta_0,n}) &= \int_{x_{-}}^{x_{+}} u_j(\beta,\omega) \, \frac{\partial}{\partial \beta} \left( n^2 V_{\sigma_1}   +   \beta |n| \widetilde{f} \omega  - \omega^2 \right)  \, u_j(\beta,\omega) \, \td x \nonumber \\ &= \int_{x_{-}}^{x_{+}} u_j(\beta,\omega) \,  \left(    |n|
	\widetilde{f} \omega   \right)  \, u_j(\beta,\omega)\, \td x.
\end{align}  Since $\widetilde{f} < 0$, we have,
\begin{align}
	\beta |n| \widetilde{f} - 2 \omega \leq -2 \omega \mbox{ for } n \in \mathbb{Z}^{-} \mbox{ and } \omega \in \mathbb{R}^{+}
\end{align}
This holds for all $ \beta \in [0,1]$ and 
we have from Lemma \ref{bound_omega} that $ \omega \geq n C_{\beta} $. Hence we have a uniform constant $ B := \inf_{ \substack{\beta \in [0,1]}} C_{\beta}$ such that 

\begin{align}
	\beta |n| \widetilde{f} - 2 \omega \leq -2 B n
\end{align}
This means that $ \di \frac{\partial \Lambda_j}{\partial \omega} (\beta_0, \omega_{\beta_0,n}) $ is bounded away from zero. By the implicit function theoreom, this proves the existence of $ \omega_{\beta,n}(\beta) $ in a neighborhood of $\beta_0$. We can compute the derivative of $\omega_{\beta,n}(\beta)$ at $\beta_0$ using, 
\begin{align}
	\di \frac{\td \omega_{\beta,n}}{\td \beta} (\beta_0) = \di -  \frac{\di \frac{\partial \Lambda_j}{\partial \beta} (\beta_0,\omega_{\beta_0,n})}{\di\frac{\partial \Lambda_j}{\partial \omega} (\beta_0,\omega_{\beta_0,n})}, \qquad
	\frac{\partial \Lambda_j}{\partial \omega} (\beta_0,\omega_{\beta_0,n}) = |n| \widetilde{f} \omega_{\beta_0,n}. 
\end{align}
Similar to the argument above, we have that $ \di \frac{\partial \Lambda_j}{\partial \omega} (\beta_0,\omega_{\beta_0,n}) $ is bounded away from zero. We hence arrive at
\begin{align}
	\label{negbeta}
	- |n|\widetilde{C}_\beta   \leq  \frac{\td \omega_{\beta,n}}{\td \beta} (\beta_0) <0
\end{align} for some $\widetilde{C}_\beta >0$ and $\beta \in [0,1]$. 
\end{proof}
\subsection{Existence of eigenvalues for the nonlinear problem} 
We conclude this section by demonstrating the existence of eigenvalues for the nonlinear problem. Note that along the same lines as Remark 8.20 in \cite{benomio:2018ivy}, it is clear from the final part of  Fig.\ref{cty_nonlinear_potential} that the energy level $E\in[1.405,1.415]$ which is an `appropriate value' for the nonlinear problem also works for the linear problem in the sense that Lemma \ref{cty_potential} holds for the chosen value of $E$.
\begin{theorem}
	\label{omega_bounds}
	Consider fixed energy levels $E$ and $ \mathcal{E}$ as in Lemma \ref{npcty}. Let $n \in \mathbb{Z}^{-}$. Given eigenvalues $\omega^2_{\mathrm{lin},n}$ for the linear eigenvalue problem where  $\omega^2_{\mathrm{lin},n}>0$, there exists an eigenvalue $\omega^2_{n}$ and corresponding eigenfunction $u_n$ to the nonlinear eigenvalue problem for  large enough $n$. Furthermore $\omega_n>0$ and the following bound holds for any $\delta >0$,
	\begin{align}
		\mathcal{C} \leq \frac{\omega_n^2}{n^2} \leq E^2 +\delta
	\end{align}
	where the constant $\mathcal{C}$ is independent of $n$.
\end{theorem}

\begin{proof}
	We start by looking at the eigenvalue problem for  $\beta=0$. We know from the linear eigenvalue problem that for large $n^2$ there exists a $\omega_{0,n}$ such that $\mathcal{Q}(0,\omega_{0,n})$ admits a zero eigenvalue i.e., $\Lambda_j(0,\omega_{0,n})=0$ for some $j$. By Lemma \ref{bound_omega}, $\omega_{0,n} \neq 0$. 
	
	Let $\omega_{0,n} >0$, then from Lemma \ref{implicit}, for some $\epsilon >0$ there exists a continuous function $\omega_{\beta,n}(\beta)$ such that for any $\beta \in  [0 ,\epsilon)$ the nonlinear eigenvalue problem admits a zero eigenvalue i.e., $ \Lambda_j(\beta,\omega_{\beta,n})=0$ for some $j$. By \eqref{negbeta}, 
	\begin{align}
		\omega^2_n = \omega^2_{1,n}(1) \leq \omega^2_{0,n}(0) \leq C n^2
	\end{align}
	Here $C$ does not depend on $n$. The bound $\omega_{0,n}(0) \leq C n^2 $ comes from conditions on appropriate energy levels $E$ from the assumptions of Lemma \ref{cty_potential}. In conjunction with Lemma \ref{bound_omega}, we have,
	
	\begin{align}
		C_1 \leq \frac{\omega^2_n}{n^2} \leq C_2
	\end{align}
	for constants $C_1$ and $C_2$ independent of $n$. For  $\beta \in [0,1]$, let us consider the problems 
	\begin{align}
		\overline{Q}_{\beta} u = E^2_j(\beta) u
	\end{align}
	where
	\begin{align}
	\begin{aligned}
		\overline{Q}_\beta u &:=  -  \frac{1}{n^2}\frac{\td^2 u}{ \td x^2 } +  \frac{1}{n^2}( n^2 V_{\sigma_1}  +   \beta |n| \widetilde{f} \omega ) \\
		 & E_j(\beta) := \frac{\omega^2_{\beta,n}(\beta)}{n^2}
			\end{aligned}
	\end{align}
	We have from  Weyl's law for the linear problem in conjunction with Lemma \ref{npcty} that, 
	\begin{align}
		E^2_j(0) \in [E^2 - \delta, E^2 + \delta]
	\end{align}
	for any arbitrary small $\delta$ and sufficiently large $n^2$. For $ n \in \mathbb{Z}^{-}$ and $\beta \in [0,1]$ we have the estimate,
	
	\begin{align}
		0 \leq \int_{x_{-}}^{x_{+}} u
		( \overline{Q}_{0} - \overline{Q}_{\beta} ) u \,\td x  = \int_{x_{-}}^{x_{+}} -|n| \widetilde{f} \omega_{\beta,n} \, \td x
	\end{align}
	which means,
	\begin{align}
		\int_{x_{-}}^{x_{+}} u \overline{Q}_\beta u \, \td x \leq \int_{x_{-}}^{x_{+}} u \overline{Q}_0 u \, \td x
	\end{align}
	implying 
	\begin{align}
		E^2_j(\beta) \leq E^2_j(0).
	\end{align}
	In particular this means, 
	\begin{align}
		E^2_j(1) \leq E^2_j(0) \leq E^2 + \delta
	\end{align}
	Combining the bounds we have 
	\begin{align}
		\begin{aligned}
			\mathcal{C} \leq E^2_j(1) \leq E^2 +\delta \\ 
			\mathcal{C} \leq \frac{\omega_n^2}{n^2} \leq E^2 +\delta
		\end{aligned}
	\end{align}	
\end{proof}
\section{Lower bound on the uniform energy decay rate}
\label{sec:lowerbound}
The purpose of this section is to prove an energy estimate for solutions to the eigenvalue problem discussed in the previous sections.  This is the main step towards establishing the desired lower bound on energy decay. We begin with the following basic lemma which can be proved using integration by parts.
\begin{lemma}
	\label{basic_identity}
	Let $ x_{-} < x_{+}$, $h >0 $ be a constant and $W $ and $\phi$ be smooth functions on $[x_{-},x_{+}]$. Then, for all smooth functions $u $  defined on $[x_{-},x_{+}]$,
	\begin{align}
		\nonumber	\int_{x_{-}}^{x_{+}} \left( \left| \frac{\mathrm{d} }{\mathrm{d} x} \left( e^{\phi/ h} u\right) \right|^2 + h^{-2} \left( W - \left(\frac{\td \phi}{\td x} \right)^2  \right) e^{2 \phi/h} |u|^2 \right) \mathrm{d} x = \int_{x_{-}}^{x_{+}} \left( -\frac{ \mathrm{d}^2 u }{\mathrm{d} r^2} + h^{-2} W u \right) u \, e^{2 \phi/h} \mathrm{d} x
	\end{align}
	
	\begin{proof}
		We start by expanding the expression on the left hand side.
		\begin{align}
			\begin{aligned}
				\int_{x_{-}}^{x_{+}} & \left( \left| \frac{\td }{\td x} \left( e^{\phi/ h} u\right) \right|^2  + h^{-2} \left( W - \left(\frac{\td \phi}{\td x} \right)^2  \right) e^{2 \phi/h} |u|^2 \right) \td x \nonumber \\ & =
				\int_{x_{-}}^{x_{+}} \left( \left| u h^{-1} e^{\phi/ h}  \frac{\td \phi }{\td x} +  e^{\phi/ h}  \frac{\td u }{\td x} \right|^2 + h^{-2} W e^{2 \phi/h} |u|^2  - h^{-2} \left(\frac{\td \phi}{\td x} \right)^2  e^{2 \phi/h} |u|^2 \right) \td x  \nonumber \\ & = \int_{x_{-}}^{x_{+}} \left(  u^2 h^{-2} e^{2\phi/ h} \left( \frac{\td \phi }{\td x}\right)^2 +  e^{2\phi/ h}  \left(\frac{\td u }{\td x}\right)^2 + \frac{2 |u| e^{2 \phi/ h}}{h}  \left(\frac{\td \phi}{\td x}\right)  \left(\frac{\td u }{\td x}\right) \right. \nonumber \\ & \qquad \qquad \left. + h^{-2} W e^{2 \phi/h} |u|^2  - h^{-2} \left(\frac{\td \phi}{\td x} \right)^2  e^{2 \phi/h} |u|^2 \right) \td x  \nonumber \\ & = \int_{x_{-}}^{x_{+}} \left(   e^{2\phi/ h}  \left(\frac{\td u }{\td x}\right)^2 + \frac{2 u e^{2 \phi/ h}}{h}  \left(\frac{\td \phi }{\td x}\right)  \left(\frac{\td u }{\td x}\right)  + h^{-2} W e^{2 \phi/h} |u|^2   \right) \td x  \nonumber \\ 
				\mbox{(using }& \mbox{integration by parts on the second term)} \nonumber \\ 
				& = \left.  e^{2 \phi/ h} u \left(\frac{\td u }{\td x}\right) \right|^{x_{+}}_{x_{-}} - \int_{x_{-}}^{x_{+}} e^{2 \phi/ h} \frac{\td }{\td x} \left( u  \left(\frac{\td u }{\td x}\right) \right) +  \int_{x_{-}}^{x_{+}} \left(   e^{2\phi/ h}  \left(\frac{\td u }{\td x}\right)^2   + h^{-2} W e^{2 \phi/h} |u|^2   \right) \td x \nonumber \\ 
				& = \int_{x_{-}}^{x^{+}} \left( -\frac{ \td^2 u }{\td x^2} + h^{-2} W u \right) u e^{2 \phi/h} \td x \nonumber
			\end{aligned}
		\end{align}
	\end{proof}
\end{lemma}

\subsection{Agmon distance}
Consider the effective potential  $$V^{h,E}_{\mbox{eff}} = h^2 V _{(En,n)} $$
where we recall the previously defined  semi-classical parameter  $h>0$ $$ h^2 = n^{-2} $$ (note that, since we have taken $n \in \mathbb{Z}^-$,  $h = -1/n$) and the energy level $E$ is chosen as in Lemma \ref{npcty}. The Agmon distance between two points $x_1$ and $x_2$ associated to the energy level $E$ and potential $V_{\mbox{eff}}$ is defined as
\begin{align}
	d(x_1,x_2) :=  \left| \int_{x_{1}}^{x_{2}}   \sqrt{V^{h,E}_{\mbox{eff}}(x)} \chi_{\left\{ V_{\mbox{eff}}^{h,E} \geq 0 \right\}}\,  \td x \right|. 
\end{align}
Physically the Agmon distance is a measure of distance between two points in the classically forbidden region. Agmon distance is the distance associated with the Agmon metric, $V_{+} \td x^2$ where $V_{+}:=\max(0,V)$. The Agmon distance satisfies the bound \cite{holzegel:2013kna}
\begin{align}
	| \nabla_x d(x,x_2)|^2 \leq \max \left\{V^{h,E}_{\mbox{eff}}(x),0 \right\}.
\end{align}
For a given $E$, using Agmon distance, one can define the distance to the classically allowed region as 
\begin{align}
	d_E(x) :=   \inf_{x_1 \in \left\{  V^{h,E}_{\mbox{eff}} \leq 0 \right\} } d(x_1,x)  .
\end{align}
We recall that $\Omega := [x_{-},x_{+}]$. For $ \epsilon \in (0,1)$, we define 
\begin{align}
	\Omega^{+}_\epsilon(E) := \left\{ x : V^{h,E}_{\mbox{eff}} > \epsilon \right\} \cap \Omega
\end{align}
with its complement in $\Omega$ defined by,
\begin{align}
	\Omega^{-}_\epsilon(E) := \left\{ x : V^{h,E}_{\mbox{eff}} \leq \epsilon \right\} \cap \Omega
\end{align}
We can now state and prove the following exponentially weighted energy estimate.
\begin{figure}[!htb]
	\includegraphics[scale=1.25]{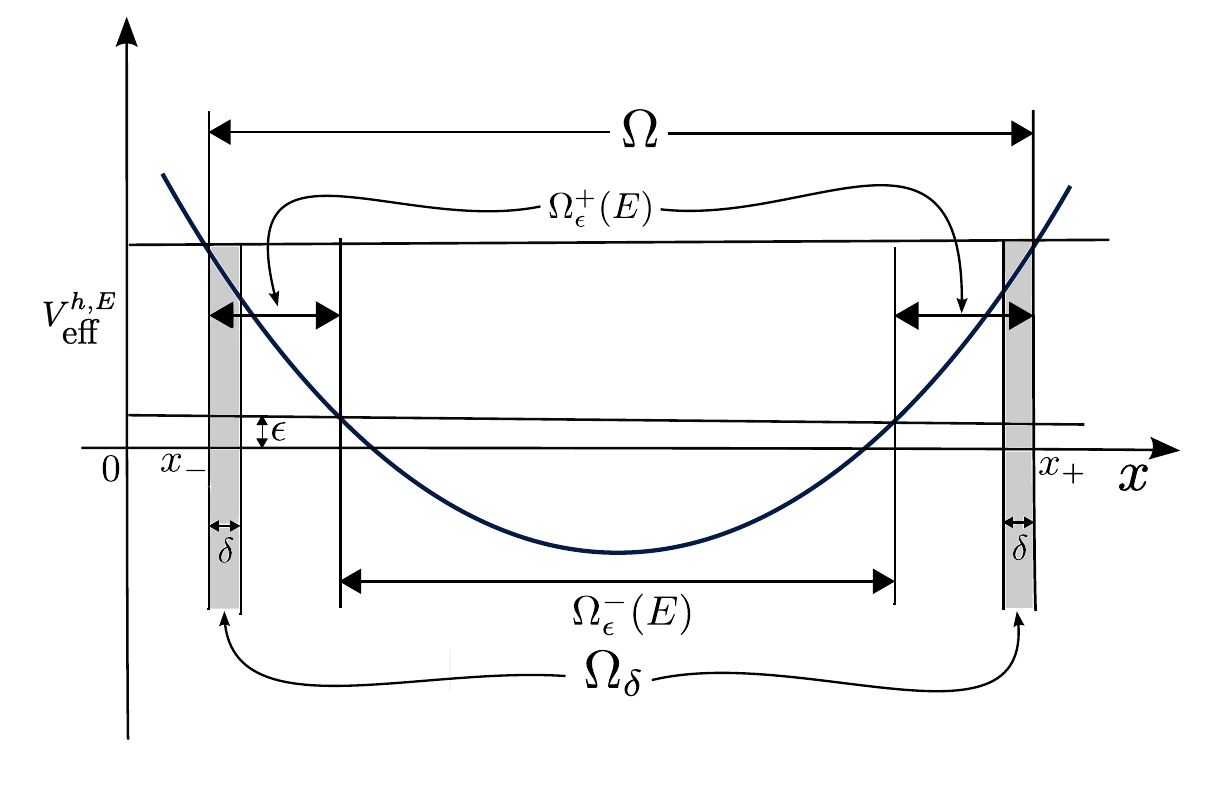}
	\caption{Classical and forbidden regions}
	\label{agmon_est}
\end{figure}

\begin{lemma}[Energy estimate]
	\label{exp_est}
	Let $u$ be a solution to the nonlinear eigenvalue problem \eqref{np_family}. Let $\kappa$ be a eigenvalue satisfying $|\kappa^2 - E^2| \leq \delta$. For $\epsilon \in (0,1)$, define 
	\begin{align}
		\phi_{E,\epsilon} (x) := (1-\epsilon ) d_{E}(x) \quad \mbox{ and } \quad  
		a_E(\epsilon) := \sup_{\Omega^{-}_\epsilon (E)} d_E .
	\end{align} Then for sufficiently small $\epsilon$ and $h$ and sufficiently small $\delta$ (depending on $\epsilon$ and $h$), $u$ satisfies
	
	\begin{align}
		\int_{\Omega} h^2  \left| \frac{\td }{\td x} \left(e^{\phi_{E,\epsilon}/h} u \right) \right|^2 \mathrm{d} x  + \frac{1}{2} \epsilon^2 \int_{\Omega^{+}_\epsilon} e^{2 \phi_{E,\epsilon}/h } |u|^2 \mathrm{d} x  \leq C \left( \kappa^2 + \frac{1}{2} \epsilon \right) e^{2 a_E(\epsilon)/h} \norm{u}^2_{L^2(\Omega)}
	\end{align}
	where the constant $C$ depends only on the parameters of the soliton spacetime and $\Omega$.
\end{lemma}

\begin{proof}
	We apply Lemma \ref{basic_identity} to $u$ with the following identifications,
	\begin{align}
		W = V^{h,\kappa}_{\mbox{eff}}\,,  \quad \phi = \phi_{E,\epsilon}
	\end{align}
	Since $u$ is a solution to the eigenvalue problem, the right hand side vanishes which gives, 
	\begin{align}	
		\int_\Omega \left( \left| \frac{\td }{\td x} \left( e^{\phi_{E,\epsilon}/ h} u\right) \right|^2 + h^{-2} \left(V^{h,\kappa}_{\mbox{eff}} - \left(\frac{\td \phi_{E,\epsilon}}{\td x} \right)^2  \right) e^{2 \phi_{E,\epsilon} /h} |u|^2 \right) \td x = 0. \end{align} This can be rewritten as,
	\begin{align}
		\int_\Omega  h^2 \left| \frac{\td }{\td x} \left( e^{\phi_{E,\epsilon}/ h} u\right) \right|^2 \td x + \int_{\Omega^+_\epsilon(E)} \left(V^{h,\kappa}_{\mbox{eff}} - \left(\frac{\td \phi_{E,\epsilon}}{\td x} \right)^2  \right) e^{2 \phi_{E,\epsilon} /h} |u|^2  \td x \nonumber \\ = \int_{\Omega^{-}_\epsilon(E)} \left( - V^{h,\kappa}_{\mbox{eff}} + \left(\frac{\td \phi_{E,\epsilon}}{\td x} \right)^2  \right) e^{2 \phi_{E,\epsilon} /h} |u|^2  \td x.
	\end{align}	 
	We make the following observations :  
	\begin{enumerate}
		\item By definition, we have $ \left. \phi_{E,\epsilon} \right|_{\Omega^{-}_\epsilon(E)} \leq  a_E(\epsilon)$.
		\item $\norm{u}^2_{L^2(\Omega^{-}_\epsilon (E))} \leq \norm{u}^2_{L^2(\Omega)} $.
		\item We have in ${\Omega^{-}_\epsilon (E)}$, 
		\begin{align}
			\left(\frac{\td \phi_{E,\epsilon}}{\td x}\right)^2  & = (1-\epsilon)^2 |\nabla_x d_{E}|^2 \nonumber \\ & \leq (1-\epsilon)^2 \epsilon \quad (\mbox{ by definition }) \nonumber \\ 
			&\leq (1-\epsilon) \epsilon \quad (\mbox{ since } \epsilon <1) 
		\end{align}
		We also note that $ - V^{h,\kappa}_{\mbox{eff}}  \leq \kappa^2 $ which finally gives, 
		\begin{align}
			- V^{h,\kappa}_{\mbox{eff}} + \left(\frac{\td \phi_{E,\epsilon}}{\td r}\right)^2 \leq \kappa^2 + \epsilon(1-\epsilon)
		\end{align}
		We can thus estimate the integral on the right hand side as, 
		\begin{align}
			\begin{aligned}
				\int_{\Omega^{-}_\epsilon(E)} \left( - V^{h,\kappa}_{\mbox{eff}} + \left(\frac{\td \phi_{E,\epsilon}}{\td x} \right)^2  \right) e^{2 \phi_{E,\epsilon} /h} |u|^2  \td x  &
				\leq (\kappa^2 + \epsilon(1-\epsilon)) \int_{\Omega^{-}_\epsilon(E)} e^{2 \phi_{E,\epsilon} /h} |u|^2  \td x \\ & \leq (\kappa^2 + \epsilon(1-\epsilon))e^{2 a_E(\epsilon)/h} \norm{u}^2_{L^2(\Omega)} \\ & \leq \left( \kappa^2 + \frac{1}{2} \epsilon \right) e^{2 a_E(\epsilon)/h} \norm{u}^2_{L^2(\Omega)} \quad \mbox{ for } \epsilon < \frac{1}{2}
			\end{aligned}
		\end{align}

		\item Consider the region $ \Omega^{+}_\epsilon (E)$. The potential $V^{h,\kappa}_{\mbox{eff}}$ is continuous in $\kappa$. Hence given a $\delta^{'} >0$ one can find a $\delta$ such that $|\kappa^2 - E^2| \leq \delta \implies |V^{h,\kappa}_{\mbox{eff}} - V^{h,E}_{\mbox{eff}}| < \delta^{'} $. Hence we have,
		\begin{align}
			V^{h,E}_{\mbox{eff}} - \delta^{'}	< V^{h,\kappa}_{\mbox{eff}} < V^{h,E}_{\mbox{eff}} + \delta^{'} 
		\end{align}
		With this we can estimate the second integrand on the left hand side as,		
		\begin{align}
			\begin{aligned}
				V^{h,\kappa}_{\mbox{eff}} - \left(\frac{\td \phi_{E,\epsilon}}{\td x} \right)^2 &\geq V^{h,E}_{\mbox{eff}} - \delta^{'} - (1-\epsilon)^2 V^{h,E}_{\mbox{eff}} \\ & \geq V^{h,E}_{\mbox{eff}} - \delta^{'} - (1-\epsilon) V^{h,E}_{\mbox{eff}} \quad (\mbox{ since } \epsilon < 1 ) \\ &\geq \epsilon V^{h,E}_{\mbox{eff}}  - \delta^{'}  \\ &\geq  \epsilon^2 - \delta^{'} \quad ( \mbox{ from the definition of } \Omega^{+}_\epsilon (E) \,) \\ 
				& \geq \frac{\epsilon^2}{2} \quad \mbox{ for any } \delta^{'} < \frac{\epsilon^2}{2}
			\end{aligned}
		\end{align}	  
	\end{enumerate}
	Hence for $\di \epsilon \leq \frac{1}{2}$ there exists a sufficiently small $\delta^{'}$ such that estimate in the theorem holds.  	
\end{proof}

Quasimodes, as explained above,  are functions that solve the wave equation everywhere except in the cut-off region. Hence, we require estimates for $u$ in the cut-off region to approximate this deviation and determine how the resulting error depends on the frequency parameter $n$. We first define the cut-off region $ \Omega_\delta$ as
\begin{align}
	\Omega_\delta := \{ x \in \Omega : \mbox{dist}(x,\partial \Omega) \leq \delta \} . 
\end{align}
The following theorem estimates $u$ on $\Omega_\delta$.

\begin{theorem}
	\label{exp_estimate}
	From Theorem \ref{omega_bounds}, we have that for $\mathcal{E},E$ and sufficiently small $\delta^{'}$, there exist eigenvalues $\kappa_n := \di \frac{\omega^2_{n}}{n^2}$ and corresponding eigenfunctions $u_n$ to the nonlinear eigenvalue problem for large enough $|n|$ such that $\omega_n >0 $ and
	\begin{align}
		\mathcal{C} \leq \kappa^2_n \leq E^2 +\delta^{'}
	\end{align}
	where the constant $\mathcal{C}$ is independent of $n$. Then for any sufficiently small $\delta$,  there holds
	\begin{align}
		\int_{\Omega_{\delta}}  \left( \left|\frac{\partial u}{\partial x} \right|^2 + |u|^2 \right) \mathrm{d} x  \leq C e^{-C |n|} \norm{u}_{L^2(\Omega)}
	\end{align}
	for a constant $C$ independent of $n$.
\end{theorem}
\begin{proof}
	From Lemma \ref{exp_est}, we have,
	\begin{align}
		\int_{\Omega} h^2  \left| \frac{\td }{\td x} \left(e^{\phi_{E,\epsilon}/h} u \right) \right|^2 \td x  + \frac{1}{2} \epsilon^2 \int_{\Omega^{+}_\epsilon} e^{2 \phi_{E,\epsilon}/h } |u|^2 \td x  \leq C \left( \kappa^2 + \frac{1}{2} \epsilon \right) e^{2 a_E(\epsilon)/h} \norm{u}^2_{L^2(\Omega)}
	\end{align}
	
	Both terms in the left hand side are positive, so the inequality applies to each, that is
	
	\begin{align}
		\int_{\Omega} h^2  \left| \frac{\td }{\td x} \left(e^{\phi_{E,\epsilon}/h} u \right) \right|^2 \td x   \leq C \left( \kappa^2 + \frac{1}{2} \epsilon \right) e^{2 a_E(\epsilon)/h} \norm{u}^2_{L^2(\Omega)} ,\label{eq1}\\
		\frac{1}{2} \epsilon^2 \int_{\Omega^{+}_\epsilon} e^{2 \phi_{E,\epsilon}/h } |u|^2 \td x  \leq C \left( \kappa^2 + \frac{1}{2} \epsilon \right) e^{2 a_E(\epsilon)/h} \norm{u}^2_{L^2(\Omega)} .
		\label{eq2}
	\end{align} Since $\Omega_\delta \subset \Omega^+_\epsilon(E)$,   \eqref{eq2} becomes
	\begin{align}
		\int_{\Omega_\delta} e^{2 \phi_{E,\epsilon}/h } |u|^2 \td x  \leq \frac{C}{\epsilon^2} \left( \kappa^2 + \frac{1}{2} \epsilon \right) e^{2 a_E(\epsilon)/h} \norm{u}^2_{L^2(\Omega)} 
	\end{align}
	With $\di \frac{\epsilon}{2} < \frac{1}{2}$ and by the definition of $\phi_{E,\epsilon}$, we see that there is a uniform constant $c$ such that, $\phi_{E,\epsilon} \geq c$ for any $x \in \Omega_{\delta}$ and $|\kappa^2 - E^2| \leq \delta^{'}$.
	By definition, we have
	\begin{align}
		a_E(\epsilon) &:= \sup_{\Omega^{-}_\epsilon (E)} \quad \inf_{x_1 \in \left\{  V^{h,E}_{\mbox{eff}} \leq 0 \right\} } d(x_1,x)  .
	\end{align} For $x \in \Omega^{-}_\epsilon(E)$, 
	\begin{align}
		\inf_{x_1 \in \left\{  V^{h,E}_{\mbox{eff}} \leq 0 \right\} } d(x_1,x)  \leq \sqrt{\epsilon} \Delta
	\end{align}
	where $\Delta = \underset{x \in \left\{  V^{h,E}_{\mbox{eff}} \leq 0 \right\}}{\max} {d(x^{\pm}_E,x)}$. Hence $ a_{E}(\epsilon) \to 0 $ as $\epsilon \to 0$ and there exists $\epsilon$ small enough such that $a_E(\epsilon) \leq c/2$. We note that $a_E(\epsilon) \to 0$ as $\epsilon \to 0$ and $\epsilon \to 0$ independently of $h \to 0$. Putting these together, we have
	\begin{align}
		e^{ 2c/h } \int_{\Omega_\delta}  |u|^2 \td x  & \leq \frac{C}{\epsilon^2} \left( \kappa^2 + \frac{1}{2} \epsilon \right) e^{2 a_E(\epsilon)/h} \norm{u}^2_{L^2(\Omega)}, \\ 
		\int_{\Omega_\delta}  |u|^2 \td x  &  \leq e^{- c/h } \frac{C}{\epsilon^2} \left( \kappa^2 + \frac{1}{2} \epsilon \right) \norm{u}^2_{L^2(\Omega)} .
	\end{align}
	There exists a constant $C$ such that, 
	\begin{align}
		\int_{\Omega_\delta} |u|^2 \leq C h^{-2} e^{-C/h} \norm{u}^2_{L^2(\Omega)} 
		\label{eq3}
	\end{align}
	Since $\epsilon \to 0$ uniformly in $h$, we can absorb $h^{-2}$ in $C$ giving  
	\begin{align}
		\int_{\Omega_\delta} |u|^2 \leq C  e^{-C/h} \norm{u}^2_{L^2(\Omega)} .
		\label{eq4}
	\end{align}
	Now, similarly, the left hand side of \eqref{eq1} becomes
	\begin{align}
	\begin{aligned} \nonumber
		\int_{\Omega_\delta} h^2  \left| \frac{\td }{\td x} \left(e^{\phi_{E,\epsilon}/h} u \right) \right|^2 \td x  &=  \int_{\Omega_\delta} h^2  \left( e^{\phi_{E,\epsilon}/h} \frac{u}{h} \frac{\td \phi_{E,\epsilon}}{\td x}  + e^{\phi_{E,\epsilon}/h} \frac{\td u}{\td x}  \right)^2 \td x \\ 
		&  = \int_{\Omega_\delta} h^2 e^{2\phi_{E,\epsilon}/h}  \left(  \frac{u}{h} \frac{\td \phi_{E,\epsilon}}{\td x}  + \frac{\td u}{\td x}  \right)^2 \td x \\  & = \int_{\Omega_\delta} h^2 e^{2\phi_{E,\epsilon}/h}  \left\{  \frac{u^2}{h^2} \left(\frac{\td \phi_{E,\epsilon}}{\td x}\right)^2  + \left(\frac{\td u}{\td x}\right)^2  + 2 \frac{u}{h} \left(\frac{\td \phi_{E,\epsilon}}{\td x}\right) \left(\frac{\td u}{\td x}\right)  \right\} \td x 
		\end{aligned}
	\end{align}
	Discarding the first term which is positive, we have
	\begin{align}
		\int_{\Omega_\delta} h^2 e^{2\phi_{E,\epsilon}/h}  \left\{ \left(\frac{\td u}{\td x}\right)^2  + 2 \frac{u}{h} \left(\frac{\td \phi_{E,\epsilon}}{\td x}\right) \left(\frac{\td u}{\td x}\right)  \right\} \td x \leq C \left( \kappa^2 + \frac{1}{2} \epsilon \right) e^{2 a_E(\epsilon)/h} \norm{u}^2_{L^2(\Omega)} 
	\end{align}
	Applying Young's inequality to the second term we get, 
	\begin{align}
		\begin{aligned}
			\int_{\Omega_\delta} h^2 e^{2\phi_{E,\epsilon}/h}  &\left\{ \left(\frac{\td u}{\td x}\right)^2 + 2 \frac{u}{h} \left(\frac{\td \phi_{E,\epsilon}}{\td x}\right) \left(\frac{\td u}{\td x}\right)  \right\} \td x   \\ & \geq \int_{\Omega_\delta} h^2 e^{2\phi_{E,\epsilon}/h}  \left\{ \left(\frac{\td u}{\td x}\right)^2 - \frac{2u^2}{h^2} \left(\frac{\td \phi_{E,\epsilon}}{\td x}\right)^2 - \frac{1}{2} \left(\frac{\td u}{\td x}\right)^2  \right\} \td x  \\ & = \int_{\Omega_\delta} h^2 e^{2\phi_{E,\epsilon}/h}  \left\{ \frac{1}{2} \left(\frac{\td u}{\td x}\right)^2 - \frac{2u^2}{h^2} \left(\frac{\td \phi_{E,\epsilon}}{\td x}\right)^2  \right\} \td x
		\end{aligned}
	\end{align}
	Hence we have,
	\begin{align}
		\begin{aligned}
			& \int_{\Omega_\delta} h^2 e^{2\phi_{E,\epsilon}/h}  \left\{ \frac{1}{2} \left(\frac{\td u}{\td x}\right)^2 - \frac{2u^2}{h^2} \left(\frac{\td \phi_{E,\epsilon}}{\td x}\right)^2  \right\} \td x  \leq C \left( \kappa^2 + \frac{1}{2} \epsilon \right) e^{2 a_E(\epsilon)/h} \norm{u}^2_{L^2(\Omega)} \\ & \int_{\Omega_\delta} h^2 e^{2\phi_{E,\epsilon}/h} \frac{1}{2} \left(  \left(\frac{\td u}{\td x}\right)^2  \right) \td x  \leq C \left( \kappa^2 + \frac{1}{2} \epsilon \right) e^{2 a_E(\epsilon)/h} \norm{u}^2_{L^2(\Omega)} + \int_{\Omega_\delta} 2u^2 \left(\frac{\td \phi_{E,\epsilon}}{\td x}\right)^2
		\end{aligned}
	\end{align}
	The second term on the right hand side can be absorbed by redefining the constant $C$. Hence, we have
	\begin{align}
		\int_{\Omega_\delta} h^2 e^{2\phi_{E,\epsilon}/h}   \left(\frac{\td u}{\td x}\right)^2  \td x & \leq C \left( \kappa^2 + \frac{1}{2} \epsilon \right) e^{2 a_E(\epsilon)/h} \norm{u}^2_{L^2(\Omega)} 
	\end{align}
	Similar to  \eqref{eq2}, we use bounds on  $a_E(\epsilon)$ and $\phi_{E,\epsilon}$ to get the following,   
	\begin{align}
		\int_{\Omega_\delta} \left(\frac{\td u}{\td x}\right)^2  \td x & \leq C h^{-2} e^{-C/h} \norm{u}^2_{L^2(\Omega)} 
	\end{align}
	Abosrbing the $h^{-2}$ in $C$, we get,
	\begin{align}
		\int_{\Omega_\delta} \left(\frac{\td u}{\td x}\right)^2  \td x & \leq C e^{-C/h} \norm{u}^2_{L^2(\Omega)}
		\label{eq5} 
	\end{align}
	Combining \eqref{eq4} and \eqref{eq5} above proves the theorem.
\end{proof}

\subsection{Quasimodes and an upper bound on the error}

Quasimodes are defined as functions $\Psi_n(t,r,\theta,\phi,\widetilde{\psi})  : {\mathcal{D}} \to \mathbb{C} $ defined by,
\begin{align}
	\Psi_n(t,r,\theta,\phi,\widetilde{\psi}) := \chi(r) e^{- i \omega_n t} e^{i n \widetilde{\psi} }  R(r) Y(\theta,\phi)
	\label{sep_fn}
\end{align}
where $\chi : \mathcal{D} \to \mathbb{R}$ is a smooth cut-off function defined by the radial function
\begin{align}
	\chi(r) = \begin{cases}
		1 \quad\quad \text{if $r \in\Omega\setminus\Omega_{\delta}$}   \\
		0 \quad\quad \text{if $r \notin\Omega$}
	\end{cases} 
\end{align}
We recall the relation between $u(r)$ and $R(r)$ and the coordinates $r,w$ and $x$ here for clarity.
\begin{align}
	R(r) = \frac{u}{r \sqrt{\hat{b}(r)} },\, \quad w = \int_{r_0}^{w} \frac{\hat{b}(s)}{s W(s)} \td s \,\, \mbox{ and }\,\, w = jx 
\end{align}
These quasimodes are clearly approximate solutions, defined to be extensions of the solutions to the wave equation on  $\Omega$ to the whole spacetime $\mathcal{D}$. They fail to solve the wave equation because of the smooth extension in the cut off region outside o which they are trivial solutions (because they vanish).  The error,  i.e. $\Box_g\Psi_n$,  is supported on  $\Omega_\delta$. The following lemma estimates the error,  which is exponentially small as $|n| \to \infty$. 

\begin{remark}
In the following sections, we will take $\Omega_\delta$ and $\Omega$ to refer to the spacetime regions : $\Omega_\delta = \Omega_{\delta} \times [0, \infty) \times (0, \pi) \times
[0, 2 \pi) \times [0, 2 \pi) \text{ and } \Omega = \Omega \times [0, \infty) \times (0,\pi) \times 
[0, 2 \pi) \times [0, 2 \pi)$ respectively. Here $[0,\infty)$ is the time domain.
\end{remark} 

\begin{lemma}
	\label{quasimode_error}
	Consider quasimodes which satisfy 
	\begin{align}
		\Box_g \Psi_n = \mathrm{err}(\Psi_n)
	\end{align}
	where $\mathrm{err}(\Psi_n)$ is the error. Then for sufficiently large $|n|, n < 0$, we have,
	\begin{align}
		\norm{\Box_g \Psi_n}_{H^k(\Sigma_t)} \leq C_k e^{-C_k |n|} \norm{\Psi_n}_{L^2(\Sigma_0)}
	\end{align}		
\end{lemma}

\begin{proof}
	For functions, $\mathcal{G}$ and $\chi$, we have,
	\begin{align}
		\Box_g( \chi \mathcal{G} ) = \chi(\Box_g \mathcal{G}) + 2 g^{\mu \nu} (\partial_\mu \chi) (\partial_\nu \mathcal{G}) + \mathcal{G} (\Box_g \chi)  
		\label{wavefnproduct}
	\end{align}
	If we let $ \mathcal{G} = e^{- i \omega_n t} e^{i n \widetilde{\psi}} R(r) Y(\theta,\phi) $, then the first term vanishes everywhere since $\mathcal{G}$ solves the wave equation in $\Omega$. Using the fact that $\chi$ is smooth and therefore controlled in $L^\infty$, we have from \eqref{wavefnproduct}, 
	\begin{align}
		\begin{aligned}
			\norm{\Box_g \Psi_n}_{L^2(\Sigma_t \cap \Omega_\delta )} & \lesssim \norm{u_n}_{H^1(\Sigma_t \cap \Omega_\delta )} \\ & \lesssim \norm{u_n}_{H^1(\Sigma_0 \cap \Omega_\delta )} 
		\end{aligned}
	\end{align} 
	Note that the $L^2$-norm of all the other eigenfunctions in \eqref{sep_fn} can be bounded. Using this with Theorem \ref{exp_estimate}, we get
	\begin{align}
		\norm{\Box_g \Psi_n}_{L^2(\Sigma_t \cap \Omega_\delta )} & \leq C e^{-C|n|}\norm{\Psi_n}_{L^2(\Sigma_0 \cap \Omega)} 
	\end{align}
	which can be written as
	\begin{align}
		\norm{\Box_g \Psi_n}_{L^2(\Sigma_t \cap \Omega_\delta )} & \leq C e^{-C|n|}\norm{\Psi_n}_{L^2(\Sigma_0 )} 
	\end{align} 
	owing to the spatial localization of quasimodes. 
	To get bounds on the higher derivatives, let us make the following observations
	\begin{enumerate}
		\item We need only be concerned with $r-$derivatives of $\Box_g (\Psi_n)$ as other eigenfunctions are bounded in $L^\infty$ (as they are analytic).
		\item $\partial_r(\Box_g \mathcal{G})$ vanishes and hence the first term is vanishes. 
		\item The second and third term contain higher derivatives  of $u$. This can be bounded using the eigenvalue equation.
	\end{enumerate}
	We hence deduce that
	\begin{align}
		\norm{\Box_g \Psi_n}_{H^k(\Sigma_t \cap \Omega_\delta )} & \leq C_k e^{-C_k |n|}\norm{\Psi_n}_{L^2(\Sigma_0 )} .
	\end{align} 
\end{proof}
\subsection{Duhamel's principle} Here we adapt the standard construction of an inhomogeneous solution to the wave equation from a homogeneous one.  Suppose  $P(t,\bx;s)$ is the solution to the following initial value problem for $t>s$.
\begin{align}
\begin{aligned}
\Box_g P(t,\bx;s)(f_1,f_2) &= 0, \\ 
P(t,\bx;s) (f_1,f_1)|_{\Sigma_s} = f_1 \,,& \,\, \partial_t P(t,\bx,;s) (f_1,f_1)|_{\Sigma_s} = f_2  
\end{aligned}
\end{align}  In other words, $P(t,\bx; s)(u_0,u_1)$ is the solution of the homogeneous wave equation with initial data $(u_0, u_1)$ prescribed on the spatial hypersurface $t=s$.  Note that it is sufficient that $u_0, u_1 \in H^1_{loc}(\Sigma)$ for a solution to exist and be unique.  Now consider the function
\begin{align}
	\Psi(t,\bx) = P(t,\bx;0)(\Phi_0,\Phi_1) + \frac{1}{2} \int_{0}^{t} P(t,\bx;s)(0,(g^{00})^{-1}F(s,\bx)) \, \td s
\end{align}
\begin{claim}
	$\Psi(t,\bx)$ solves the following initial value problem (inhomogeneous wave equation).	
	\begin{align}
	\begin{aligned}
	\Box_g \Psi (t,\bx)  &= F(t,\bx) \\ \Psi(\bx,0) = \Phi_0(\bx) , \quad \, & \, \partial_t \Psi (0,\bx) = \Phi_1(\bx)
	\end{aligned}
	\end{align}	
\end{claim}

\begin{proof} 
	\begin{align}
		\begin{aligned}
			\Box_g \Psi = & \, g^{a0} (\partial_a \partial_0 \Psi) + g^{ai}(\partial_a \partial_i \Psi) - g^{ab} \Gamma^0_{ab} \partial_0 \Phi - g^{ab} \Gamma^i_{ab} \partial_i \Psi \nonumber \\ 
			= & \, \frac{g^{00}}{2} \partial_t P(t,\bx;t) + \frac{g^{0i}}{2} \partial_i P(t,\bx;t) + \frac{g^{00}}{2} \partial_t P(t,\bx;t)  + \frac{g^{00}}{2} \int_{0}^{t} \partial^2_t P \td s + \frac{g^{0i}}{2} \int_{0}^{t} \partial_i \partial_0 P \td s  \nonumber \\ &+ \frac{g^{0i}}{2} \partial_i P(t,\bx; t) + \frac{g^{0i}}{2} \int_{0}^{t} \partial_0 \partial_i P \td s + \frac{g^{ij}}{2} \int_{0}^{t} \partial_i \partial_j P \td s - \frac{1}{2} \Gamma^0_{ab} g^{ab} P(t,\bx;t) - \frac{1}{2} g^{ab} \Gamma^c_{ab} \int_{0}^{t} \partial_c P \td s  \\  = & \, g^{00} \partial_t P(t,\bx;t) + g^{0i} \partial_i P(t,\bx;t) - 
			\frac{1}{2} g^{ab} \Gamma^0_{ab} P(t,\bx;t)  + \frac{1}{2} \int_{0}^{t} \Box_g P\left(0,(g^{00})^{-1} F(s,\bx)\right) \td s \\ = & \, F(t,\bx) , 
		\end{aligned}
	\end{align}
	where we used that
	$ \Box_g P = 0 \,, \, \, P(t,\bx;t) =0 \,, \, \mbox{ and } \,\, \partial_t P(t,\bx;t) = (g^{00})^{-1} F(t,\bx) . 
	$  \end{proof}
\subsection{Bound on the uniform decay rate}  We have constructed quasimodes, namely, approximate solutions to the wave equation  $\Box_g \Psi_n = \textrm{err}_n (\Psi_n)$ with compactly supported  initial data $(\Psi_n(0,\bx), \partial_t \Psi_n(0,\bx)$. We have also seen that the error can be made exponentially small as $ |n| \to \infty$. Now consider a solution of the homogeneous wave equation with the same initial data
\begin{equation}
	\Box \Phi_n=0 , \qquad \Phi_n(0,\bx) = \Psi_n(0,\bx), \qquad \partial_t \Phi_n(0,\bx) = \partial_t \Psi_n(0,\bx). 
\end{equation}  Using Duhamel's principle we have
\begin{equation}
	\Psi_n(t,\bx) = \Phi_n(t,\bx) + \frac{1}{2} \int_0^t P(t,\bx; s)(0, (g^{00})^{-1} \textrm{err}_n(\Psi_n)) \; \td s
\end{equation} where $P(t,\bx;s)$ is a solution to the homogeneous wave equation described above.  In terms of the `local' energy integral $E_{t,\Omega}[\Phi]$  measured over $\Omega$ (recall this is quadratic in derivatives of $\Phi$)
\begin{equation}
	E_{t,\Omega}[\Psi_n - \Phi_n] = E_{t,\Omega}\left[\int_0^t P(t,\bx; s) \; \td s \right]
\end{equation} We use the fact that $P(t,\bx;t)=0$ and 
\begin{equation}
	\int_0^t \partial_\alpha P(t,\bx; s)(0, (g^{00})^{-1} \textrm{err}_n(\Psi_n)) \; \td s \leq  t \sup_{s \in [0,t]} |  \partial_\alpha P(s,\bx; s)(0, (g^{00})^{-1} \textrm{err}_n(\Psi_n)) |
\end{equation}
to get, 
\begin{equation}
	\left( E_{t,\Omega} [\Psi_n - \Phi_n] \right )^{1/2} \leq \frac{t}{2} \sup_{s \in [0,t]} \left(E_{t,\Omega}[P] \right)^{1/2}  \leq C t \left(E_{0,\Omega}[P] \right)^{1/2} 
\end{equation} where we used the uniform boundedness of the energy to express the estimate in terms of the energy at $t=0$.  Evaluating the energy of $P(t,\bx;s)$ at $t=0$, we see that
\begin{equation}
	\left(E_{0,\Omega}[P] \right)^{1/2}  \sim \norm{ (g^{00})^{-1} \textrm{err}_n (\Psi_n(0)) }_{L^2(\Omega)} \leq C e^{-C|n|} \norm{ \Psi_n(0)}_{L^2(\Omega)}
\end{equation} where we used the above estimate.  Applying the Poincare inequality we arrive at
\begin{equation}
	\left( E_{t,\Omega} [\Psi_n - \Phi_n] \right )^{1/2} \leq C t e^{-C|n|} \left(E_{0,\Omega}[\Psi_n]\right)^{1/2}. 
\end{equation}  Using the reverse triangle inequality we find
\begin{equation}
	\left| \left(E_{t,\Omega}[\Psi_n]\right)^{1/2} - \left(E_{t,\Omega}[\Phi_n]\right) ^{1/2}\right |  \leq \left(E_{t,\Omega}[\Psi_n - \Phi_n]\right)^{1/2}.
\end{equation} Therefore for all $t \leq \tfrac{1}{2C} e^{C|n|}$ there holds
\begin{equation}
	\frac{1}{2}\left(E_{0,\Omega}[\Psi_n]\right)^{1/2} \leq \left(E_{t,\Omega}[\Phi_n]\right)^{1/2}.
\end{equation} Of course, since by construction $\Psi_n$ vanishes outside of $\Omega$, we can write this as
\begin{equation}
	\left(E_{t,\Omega}[\Phi_n]\right)^{1/2} \geq \frac{1}{2}\left(E_{0}[\Psi_n]\right)^{1/2}.
\end{equation} We now bound the energy of the homogeneous solution $\Phi_n$ from below by a higher order energy.  Note that 
\begin{equation}
	E_2[\Psi_n](0) = E[\Psi_n](0) + \sum_{\alpha=0}^4 E[\partial_\alpha \Psi_n]
\end{equation} Taking derivatives with respect to $t, \widetilde{\psi}$ will simply pull down $C n, n$ respectively (since $\omega = C n$ for some $C$). Derivatives with respect to $\theta, \phi$ will simply yield linear combinations of the charged spherical harmonics $Y$.  Thus all the energies associated to these values of $\alpha$ will atmost be of order  $E[\Psi_n]$.  However,  $e_1 (\partial_r \Psi_n) \sim Y(\theta, \phi) e^{-\omega t + i n \widetilde\psi} u''$. Using the equation satisfied by $u$ we can rewrite $u'' = V(r) u$ . On the other hand $e_1(\Psi) \sim e^{-\omega t + i n \widetilde\psi} u'$.  Using the Poincare inequality we know $\norm{u}_{L^2(\Omega)} \leq C \norm{u'}_{L^2(\Omega)}$.  These considerations imply
\begin{equation}
	E_2[\Psi_n](0) \leq ( C_1 + C_2 n^2 )E[\Psi_n](0)
\end{equation} and in particular for $|n|$ sufficiently large
\begin{equation}
	\left(E[\Psi_n]\right)^{1/2} \geq \frac{C}{|n|} \left(E_2 [ \Psi_n](0)\right)^{1/2}
\end{equation} Similar inequalities will apply with $(n, E_2[\Psi_n](0))$ 
replaced with $(n^{k-1}, E_k[\Psi_n](0))$ for $k>2$ (essentially, additional derivatives will pull down factors of $n$.  Now because by construction $\Phi_n$ has the same initial data as  $\Psi_n$, we have for sufficiently large $|n|$ and  $0 < t \leq e^{C|n|}/ 2C$, $k > 2$,
\begin{equation}
	\left(E_{\Omega}[\Phi_n](t)\right)^{1/2} \geq \frac{C}{|n|^{k-1}} (E_k[\Phi_n](0))^{1/2}. 
\end{equation} The above result prevents the possibility of a local uniform logarithmic decay statement of the form
\begin{equation}
	\limsup_{t\to \infty} \, \delta(t) E_\Omega[\Phi](t) \leq C E[\Phi](0)
\end{equation} where $\delta(t)$ encodes the rate of decay, for all solutions $\Phi$ to the wave equation with suitably regular initial data. Here `uniform' implies that such a decay must hold for any smooth solution.   Setting $\tau_n = e^{C|n|}/2C$,  we obtain 
\begin{equation}
	(\log (2 + \tau_n))^2 \frac{E_\Omega[\Phi_n](\tau_n)} {E_2 [ \Phi_n](0)} \geq C
\end{equation} This produces a sequence $\{(\tau_n, \Phi_n)\}$, $|n| \geq N$ for $N$ sufficiently large, of solutions to the homogeneous wave equation.  We conclude that for some positive constant $C$, 
\begin{equation}
	\limsup_{\tau \to \infty} \sup_{\Phi \neq 0} (\log (2 + \tau))^2 \frac{E_\Omega[\Phi](\tau)} {E_2 [ \Phi](0)} \geq C 
\end{equation} where the supremum is taken over all $\Phi$ that lie in the completion of the set of smooth solutions to the wave equation with compactly supported initial data on the hypersurface $\Sigma_0$, with respect to the norm defined by $E_2$.   An analogous statement holds for $k  \in \mathbb{N}$, there are positive constants $C_k$ such that
\begin{equation}
	\limsup_{\tau \to \infty} \sup_{\Phi \neq 0} (\log (2 + \tau))^{2k} \frac{E_\Omega[\Phi](\tau)} {E_{k+1} [ \Phi](0)} \geq C_k. 
\end{equation} This completes the proof of Theorem \ref{key_theorem}. 
As emphasized by Keir \cite{Keir:2016azt}, it should be noted that the above results do not prevent smooth solutions from decaying faster than logarithmically. It simply states that there will always exist solutions that decay slower. 

\def\bibfont{\scriptsize}
\bibliography{irnew}

\end{document}